\documentclass[journal]{IEEEtran}
\usepackage{amsmath,amsfonts}
\usepackage{algorithmic}
\usepackage{algorithm}
\usepackage{array}
\usepackage[caption=false,font=normalsize,labelfont=sf,textfont=sf]{subfig}
\usepackage{textcomp}
\usepackage{stfloats}
\usepackage{url}
\usepackage{verbatim}
\usepackage{graphicx}
\usepackage{multirow}
\usepackage{cite}
\usepackage{amssymb}
\usepackage{amsthm}
\usepackage{booktabs}
\usepackage{cuted}
\usepackage{graphicx}
\newtheorem{theorem}{Theorem}

\newtheorem{corollary}{Corollary}

\newtheorem{lemma}{Lemma}

\newtheorem{proposition}{Proposition}
\newtheorem{remark}{Remark}

\def\m #1{\boldsymbol{#1}}

\def\bee{\begin{equation}}
	\def\ene{\end{equation}}

\def\beq{\begin{eqnarray}}
	\def\enq{\end{eqnarray}}

\begin{document}

\title{Deterministic Cram\'{e}r-Rao Bounds for Coherent Direction-of-Arrival Estimation: Rank Information Versus Coherence Structure}
\author{Weichao~Zheng and Zai~Yang
	
	\thanks{The authors are with the School of Mathematics and Statistics,
		Xi’an Jiaotong University, Xi’an 710049, China, E-mail: zwcabc@stu.xjtu.edu.cn; yangzai@xjtu.edu.cn. (\textit{Corresponding author: Zai Yang.})}}
\maketitle

\begin{abstract}
    Direction-of-arrival (DOA) estimation is a fundamental problem in array signal processing, for which the Cram\'{e}r–Rao bound (CRB) serves as a standard performance benchmark under both stochastic and deterministic source models. In multipath environments, the source matrix is of low-rank and, more specifically, exhibits a within-group proportionality structure. Although stochastic CRBs for coherent DOA estimation have been studied, theoretical results for their deterministic counterparts remain limited. This paper aims to bridge this gap. We first derive a tangent-space condition characterizing when structural constraints on the source matrix can strictly reduce the frequency CRB. We then prove that, for uniform linear arrays, imposing only the low-rankness of the source matrix leaves the frequency CRB unchanged, although it reduces the source-matrix CRB whenever the rank constraint is nontrivial, whereas fully exploiting the coherence structure yields a strictly smaller frequency CRB for almost all parameter values under mild conditions. These conclusions are further extended to analytic array manifolds and multidimensional DOA models, with componentwise results established for uniform planar arrays. The results demonstrate that performance gains in DOA estimation arise from the detailed coherence structure rather than low-rankness alone. Numerical experiments are finally provided to validate the theoretical findings.
\end{abstract}


\section{Introduction}
Direction-of-Arrival (DOA) estimation is a central problem in array signal processing, whose goal is to estimate the directions of multiple sources impinging on an array, based on a series of snapshots of the array output \cite{stoica2005spectral}. It has found broad applications, including radar \cite{nion2010tensor}, sonar \cite{cox1989fundamentals}, wireless communications \cite{tsai2018millimeter}, while related estimation techniques has also been applied to modal analysis \cite{li2018atomic}. As a lower bound on the error covariance of any unbiased estimator, Cram\'{e}r-Rao bound (CRB) has been widely adopted as a fundamental performance benchmark for DOA estimators under both deterministic and stochastic models \cite{stoica1989music,stoica1990performance}, where the source signals are treated as unknown deterministic parameters or random processes characterized by their statistical properties, respectively. Related CRB analyses have been developed for various array geometries \cite{liu2017cramer,wang2016coarrays,hua1991note}, which provides theoretical guidance for algorithm development \cite{stoica1989music}, performance evaluation \cite{stoica1990performance}, and array design \cite{gazzah2006cramer}.\par
In practical propagation environments, multipath reflections (e.g., from buildings, the ground, sea surfaces, or ionospheric refraction) give rise to coherent sources impinging on the array, which invalidates many conventional DOA estimators \cite{krim2002two}. This challenge has motivated extensive research, including algorithm development and CRB analyses \cite{shan1985spatial,ye2007doa,stoica1996maximum,zhang2009estimation,xie2017source,choi2000maximum}. Among these studies, some exploit only the low-rank structure of the source matrix induced by coherence \cite{shan1985spatial,ye2007doa,stoica1996maximum}, whereas others incorporate the detailed coherence relationships among the source components \cite{zhang2009estimation,xie2017source,choi2000maximum}.
Intuitively, exploiting more structural information is expected to yield a lower CRB. However, it is shown in \cite{stoica1996maximum} that, under the stochastic model, the CRB obtained by exploiting only the low-rank structure of the source matrix is identical to that obtained without exploiting this structure, with a similar result established in \cite{xie2017source} for the special case involving a single direct path. Moreover, numerical results in \cite{xie2017source} indicate that exploiting the coherent source structure yields a strictly lower CRB, although a rigorous proof of this relationship remains unavailable. These findings suggest that, under the stochastic model, potential performance gains may arise from the detailed coherence structure rather than from low-rankness alone.
\par
For the deterministic model, however, the source samples are nuisance parameters whose number grows with the snapshots, so the rank and coherence constraints act on a different parameter space and their effects cannot be directly inferred from the stochastic theory.
Deterministic methods retrieve the DOAs directly from the array output matrix and generally remain applicable in the presence of coherent sources \cite{hyder2010direction,bhaskar2013atomic,yang2016exact}. Consequently, the specific impact of source coherence under deterministic modeling has received limited attention. Nevertheless, the applicability of deterministic methods does not imply that source coherence is irrelevant to the algorithm performance limits. It is therefore essential to reveal whether and how the additional structures can improve the estimation accuracy. A natural conjecture is that results analogous to those established under the stochastic model also hold under the deterministic setting. Fig. \ref{fig1} shows a simulation result that provides empirical evidence for this conjecture. The CRBs obtained without and with exploiting the low-rank structure, denoted by CRB-G and CRB-LR, respectively, coincide with each other. At $10$ dB, their relative difference is only $9.0370\times 10^{-16}$, which can be attributed to numerical round-off errors. By contrast, the CRB obtained by exploiting the coherence structure, denoted by CRB-C, is strictly lower than CRB-G. Despite this empirical evidence, these relationships still require rigorous theoretical justification.
\par
 \begin{figure}[htbp] 
 	\centering 
 	\includegraphics[width=0.9\linewidth]{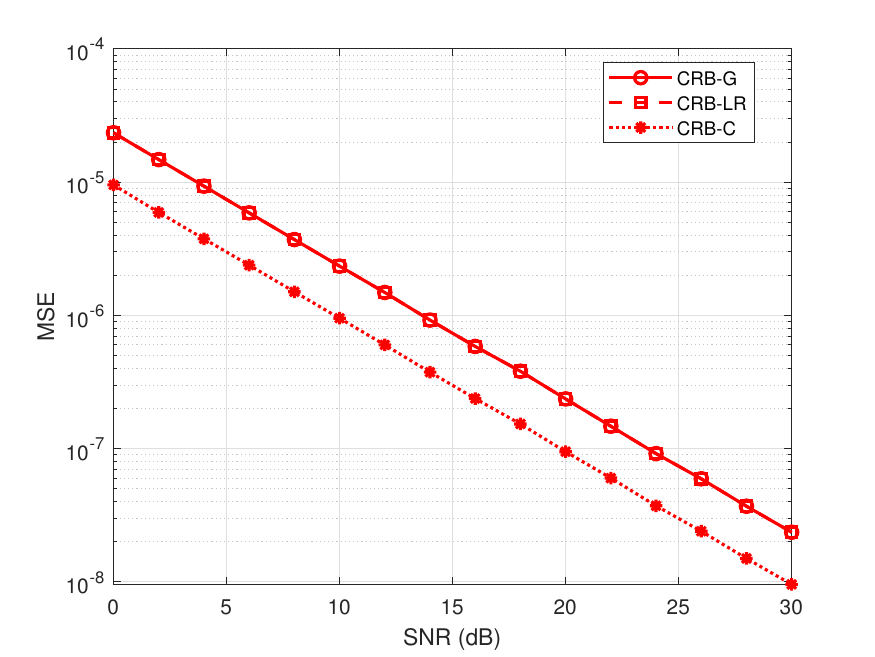} 
 	\caption{CRB comparison under the general, low-rank, and coherence models for a ULA with $M=20$ sensors and $L=20$ snapshots.} 
 	\label{fig1} 
 \end{figure}
   \begin{table*}[t]
 	\centering
 	\caption{Comparison with related works. CRB-G, CRB-LR, and CRB-C denote the CRBs obtained under the general model, the low-rank model, and the coherence model, respectively.}
 	\footnotesize
 	\label{Table I}
 	\setlength{\tabcolsep}{3.5pt}
 	\renewcommand{\arraystretch}{1.15}
 	\begin{tabular}{cccc}
 		\toprule
 		Work & Modeling on the sources & Source structure exploited & Results on CRB relationship \\
 		\midrule
 		\cite{stoica1996maximum} & Stochastic & Low-rank structure & $\textnormal{CRB-G}=\textnormal{CRB-LR}$ (theoretical result) \\
 		\cite{xie2017source} & Stochastic & Coherence structure & $\textnormal{CRB-G}>\textnormal{CRB-C}$ (numerical result) \\
 		\cite{choi2000maximum} & Deterministic & Coherence structure & $\textnormal{CRB-G}>\textnormal{CRB-C}$ (numerical result)  \\
 		Proposed & Deterministic & Coherence structure & $\textnormal{CRB-G}=\textnormal{CRB-LR}>\textnormal{CRB-C}$ (theoretical result) \\
 		\bottomrule
 	\end{tabular}
 	
 \end{table*}
In this paper, under the deterministic signal model, we rigorously establish the relationships among the CRBs corresponding to the general, low-rank, and coherence models. The main contributions are summarized as follows.
\begin{itemize}
	\item We recast the exploitation of additional structural information as imposing constraints on the source matrix and establish a necessary and sufficient condition under which constraints on a subset of parameters affect the CRB for the remaining parameters. These results may be of independent interest.
	\item Based on the above result, we rigorously prove that, for uniform linear arrays (ULAs), exploiting the low-rank structure of the source matrix does not change the frequency CRB, whereas exploiting the detailed coherence structure can yield a strictly smaller CRB.
	\item We extend the CRB results from ULAs to general one-dimensional array geometries and further to multidimensional arrays, such as planar arrays. These extensions demonstrate that the established relationships among the three CRBs are not specific to a particular array structure, but reflect more fundamental properties of deterministic coherent DOA estimation.
\end{itemize}
These results demonstrate that exploiting low-rankness alone does not yield an intrinsic performance gain, whereas fully exploiting the coherence structure can improve the achievable estimation accuracy. To the best of our knowledge, these relationships have not been rigorously proved in the existing literature. Table \ref{Table I} compares the present study with related works and highlights its main distinctions and theoretical contributions.\par 
The remainder of this paper is organized as follows. Section \ref{Signal Model} introduces the problem formulation and presents the three source models considered in this paper, namely the general, low-rank, and coherence models. Section \ref{Relationships} establishes the relationships among the CRBs under the three models, and provides rigorous proofs of these results. Section \ref{extensions} extends the conclusions to general array geometries and multidimensional DOA estimation problems. Finally, Section \ref{validation} presents numerical experiments to validate the theoretical results.

\par
Notations used in this paper are as follows. Boldface lowercase and uppercase letters denote vectors and matrices, respectively. The complex conjugate, transpose,
conjugate transpose, inverse, vectorization, Frobenius norm, column space and pseudoinverse of a matrix $\m{X}$ are denoted as $\overline{\m{X}},\m{X}^{T},\m{X}^{H},\m{X}^{-1}$, $\textnormal{vec}(\m{X}),\|\m{X}\|_{F}, \textnormal{col}(\m{X}),\m{X}^{\dagger}$, respectively. The real and imaginary part of a matrix $\m{X}$ are given by $\Re(\m{X}),\Im(\m{X})$, respectively. For a matrix $\m{X}$, its $i$-th row, $j$-th column and $(i,j)$-th element are denoted as $\m{X}_{i,:},\m{X}_{:,j},X_{ij}$, respectively. For index sets $\mathcal{I}_{1}$ and $\mathcal{I}_{2}$, $\m{X}_{\mathcal{I}_{1},\mathcal{I}_{2}}$ denotes the submatrix of $\m{X}$ with row indices in $\mathcal{I}_{1}$ and column indices in $\mathcal{I}_{2}$. For matrices $\m{X}$ and $\m{Y}$, $\m{X}\otimes\m{Y}$ denotes their Kronecker product. For a partitioned matrix $\m{X}=\begin{bmatrix}
\m{X}_{11}&\m{X}_{12}\\ \m{X}_{21}&\m{X}_{22}
\end{bmatrix}$, the Schur complement of $\m{X}_{11}$ in $\m{X}$ is given by $\textnormal{Sch}(\m{X},\m{X}_{11})=\m{X}_{22}-\m{X}_{21}\m{X}_{11}^{-1}\m{X}_{12}$. For matrices $\m{A}$ and $\m{B}$, $\m{A}\succeq \m{B}$ means that $\m{A}-\m{B}$ is positive semidefinite. The all-zero matrix of size $M\times N$ and the identity matrix of size $N$ are given by $\m{0}_{M\times N},\m{I}_{N}$, respectively. Given matrices $\m{X}_{1},\m{X}_{2},...,\m{X}_{D}$, $\textnormal{blkdiag}\left(\m{X}_{1},\m{X}_{2},...,\m{X}_{D}\right)$ denotes the block diagonal matrix with $\m{X}_{1},\m{X}_{2},...,\m{X}_{D}$ on its diagonal. For a complex-valued matrix $\m{X}$, we define
$
\mathcal{R}(\m{X})
=
\begin{bmatrix}
\Re(\m{X}) & -\Im(\m{X}) \\
\Im(\m{X}) & \Re(\m{X})
\end{bmatrix}
$ and $\mathcal{V}(\m{X})=\begin{bmatrix}
	\Re(\m{X})\\
	\Im(\m{X})
\end{bmatrix}$.
\section{Signal Model and Preparations}  \label{Signal Model}
\subsection{Signal Model}
Consider $D$ far-field narrowband sources $\{c_{d}(t)\}_{d=1}^{D}$ impinging on an $N$-element ULA with half-wavelength inter-element spacing from directions $\{\theta_{d1}\}_{d=1}^{D}$. Under a multipath propagation scenario, the $d$-th signal generates $P_{d}$ components arriving at the ULA from directions $\{\theta_{dp}\}_{p=1}^{P_{d}}$. These components are coherent with the reference signal $c_{d}(t)$, and the corresponding coherence coefficients, also known as fading coefficients, are denoted by $\{\rho_{dp}\}_{p=1}^{P_{d}}$. To avoid ambiguity between the direct-path component of one source and the multipath component of another, we assume that the source signals $\{c_{d}(t)\}_{d=1}^{D}$ are pairwise incoherent, i.e., no two rows of $\m{C}$ are proportional. Let $K=\sum_{d=1}^{D}P_{d}$
denote the total number of incident components. The array output is an $N\times 1$ vector, given by 
\begin{equation}
	\begin{aligned}
		\m{y}(t_{l})&=\sum_{d=1}^{D}\sum_{p=1}^{P_{d}}\m{a}(f_{dp})\rho_{dp}c_{d}(t_{l})+\m{e}(t_{l}),\quad l=1,2,...,L, \label{2.1.1}
	\end{aligned}
\end{equation}
where $L$ denotes the number of snapshots, $f_{dp}\in[-\frac{1}{2},\frac{1}{2})$ has a one-to-one relation to $\theta_{dp}$ by $f_{dp}=\frac{1}{2}\sin\theta_{dp}$, $\m{a}(f)$ denotes the steering vector of $f$, given by
\begin{equation}
	\begin{aligned}
		\m{a}(f)=\left[1,e^{i2\pi f},...,e^{i2\pi(N-1)f}\right]^{T}, \label{2.1.2}
	\end{aligned}
\end{equation}
and $\m{e}(t)$ denotes the Gaussian noise vector. We rewrite (\ref{2.1.1}) into a compact form as
\begin{equation}
	\begin{aligned}
		\m{Y}=\m{A}(\m{f})\m{\Psi}\m{C}+\m{E}, \label{2.1.3}
	\end{aligned}
\end{equation} where $\m{A}(\m{f})=[\m{a}(f_{1}),\m{a}(f_{2}),...,\m{a}(f_{K})]\in\mathbb{C}^{N\times K}$ denotes the steering matrix, and $\m{f}=\left[\m{f}_{1}^{T},\m{f}_{2}^{T},...,\m{f}_{D}^{T}\right]^{T}$ denotes the frequency vector with $\m{f}_{d}$ composed of $\{f_{dp}\}_{p=1}^{P_{d}}$; $\m{\Psi}=\textnormal{blkdiag}\left(\m{\rho}_{1},\m{\rho}_{2},...,\m{\rho}_{D}\right)\in\mathbb{C}^{K\times D}$ denotes the fading coefficient matrix with $\m{\rho}_{d}$ composed of $\{\rho_{dp}\}_{p=1}^{P_{d}}$; $\m{C}\in\mathbb{C}^{D\times L}$ denotes the direct-path source signal matrix, with its $(d,l)$-element given by $c_{d}(t_{l})$; and $\m{E}=\left[\m{e}(t_{1}),\m{e}(t_{2}),...,\m{e}(t_{L})\right]\in\mathbb{C}^{N\times L}$ denotes the noise matrix. 
\subsection{Three Models and Model Parameters}
The source matrix $\m{S}=\m{\Psi}\m{C}$ possesses additional structural information. First, it is low rank, with $\textnormal{rank}(\m{S})=R\leq K$, where $R=\textnormal{rank}(\m{C})\leq D$. Second, certain rows of $\m{S}$ are mutually coherent due to the multipath propagation structure.
 To examine how exploiting different levels of structural information in the source matrix affects parameter estimation performance, we consider three formulations:
 a general model that imposes no additional structure on $\m{S}$, a low-rank model that exploits only the low-rank property of $\m{S}$, and a coherent model that fully leverages the multipath-induced coherence structure. The details are as follows.
\subsubsection{General Model} We first consider the general model, where no additional information of $\m{S}$ is exploited. The model parameter vector is  $\left[\m{f}^{T},\m{\theta}_{\m{S}}^{T}\right]^{T}\in\mathbb{R}^{K+2KL}$, where $\m{\theta}_{\m{S}}$ collects the real-valued parameters in $\m{S}$, i.e., 
\begin{equation}
	\begin{aligned}
		\m{\theta}_{\m{S}}=\Big[\Re(\textnormal{vec}(\m{S}))^{T}, \Im(\textnormal{vec}(\m{S}))^{T}\Big]^{T}\in\mathbb{R}^{2KL}. \label{2.2.1}
	\end{aligned}
\end{equation}
\subsubsection{Low-Rank Model} In this case, we consider a more structured model, where the low-rank property of the source matrix $\m{S}$ is exploited. Throughout the analysis of the low-rank model for coherent sources, we assume $R<\min\{K,L\}$. When $R=\min\{K,L\}$, the rank constraint is trivial and the low-rank model coincides with the general model. We consider the following rank-$R$ decomposition
\begin{equation}
	\begin{aligned}
		\m{S}=\m{\Phi}\m{M}, \label{2.2.2}
	\end{aligned}
\end{equation}
where $\m{\Phi}\in\mathbb{C}^{K\times R}$ and $\m{M}\in\mathbb{C}^{R\times L}$. Without loss of generality, possibly after a permutation of the columns of $\m{S}$, the first $R$ columns of $\m{S}$ are assumed to be linearly independent. We then set the first $R$ columns of $\m{M}$ to be the identity matrix $\m{I}_{R}$, which removes the factorization ambiguity and ensures its uniqueness. Specifically, we have $\m{M}=\left[\m{I}_{R},\tilde{\m{M}}\right]$. As this operation only amounts to a re-parameterization of the same signal model, it does not affect the resulting CRB. The model parameters in this case are $\left[\m{f}^{T},\m{\theta}_{\m{\Phi}}^{T},\m{\theta}_{\tilde{\m{M}}}^{T}\right]^{T}\in\mathbb{R}^{K+2KR+2RL-2R^{2}}$, with 
\begin{equation}
	\begin{aligned}
		\m{\theta}_{\m{\Phi}}&=\left[\Re\left(\textnormal{vec}(\m{\Phi})\right)^{T},\Im\left(\textnormal{vec}(\m{\Phi})\right)^{T}\right]^{T}\in\mathbb{R}^{2KR}, \\
		\m{\theta}_{\tilde{\m{M}}}&=\left[\Re\left(\textnormal{vec}\left(\tilde{\m{M}}\right)\right)^{T},\Im\left(\textnormal{vec}\left(\tilde{\m{M}}\right)\right)^{T}\right]^{T}\in\mathbb{R}^{2RL-2R^{2}}. \label{2.2.3}
	\end{aligned}
\end{equation}
The rank-$R$ structure of $\m{S}$ is then equivalent to restricting the parameter vector $\m{\theta}_{\m{S}}$ to the constrained set
\begin{equation}
	\begin{aligned}
		\m{\theta}_{\m{S}}\in \mathcal{M}_{1}\triangleq \left\{\m{\theta}_{\m{S}}:\m{\theta}_{\m{S}}=\m{g}_{1}\left(\m{\theta}_{\m{\Phi}},\m{\theta}_{\tilde{\m{M}}}\right)\right\}, \label{2.2.4}
	\end{aligned}
\end{equation}
where $\m{g}_{1}$ is a smooth mapping determined by \eqref{2.2.2}. Moreover, due to the uniqueness of the factorization in \eqref{2.2.2}, the Jacobian matrix of $\m{g}_{1}$, denoted by $\m{J}_{1}$, is of full column rank.
\subsubsection{Coherence Model} We now consider the coherence model, which fully exploits the coherence structure of the source matrix $\m{S}$. We consider the following decomposition
\begin{equation}
	\begin{aligned}
		\m{S}=\m{\Psi}\m{C},\label{2.2.5}
	\end{aligned}
\end{equation}
where $\m{\Psi}=\textnormal{blkdiag}(\m{\rho}_{1},\m{\rho}_{2},...,\m{\rho}_{D})\in\mathbb{C}^{K\times D}$ and $\m{C}\in\mathbb{C}^{D\times L}$. To remove the scaling ambiguity, for each $d$, we choose an index $l_{d}\in\{1,2,...,L\}$ such that $C_{d,l_{d}}\neq 0$. Such an index always exists; otherwise, the $d$-th row of $\m{C}$ would be identically zero. We then absorb $C_{d,l_{d}}$ into $\m{\rho}_{d}$ such that $C_{d,l_{d}}=1$. This normalization removes the remaining scaling ambiguity and thereby ensures the uniqueness of the factorization. Foe each $l$, define $\mathcal{C}_{l}=\{d:l_{d}\neq l\}$, and collect the free entries in the $l$-th column of $\m{C}$ as $\tilde{\m{c}}_{l}=\m{C}_{\mathcal{C}_{l},l}\in\mathbb{C}^{|\mathcal{C}_{l}|}$. Since one entry in each row of $\m{C}$ is fixed at 1, the total number of free complex-valued entries is $\sum_{l=1}^{L}|\mathcal{C}_{l}|=DL-D$. The model parameters in this case are $\left[\m{f}^{T},\m{\theta}_{\m{\rho}_{1}},\m{\theta}_{\m{\rho}_{2}},...,\m{\theta}_{\m{\rho}_{D}},\m{\theta}_{\tilde{\m{c}}_{1}},\m{\theta}_{\tilde{\m{c}}_{2}},...,\m{\theta}_{\tilde{\m{c}}_{L}}\right]^{T}\in\mathbb{R}^{3K+2DL-2D}$, with
\begin{equation}
	\begin{aligned}
		\m{\theta}_{\m{\rho}_{d}}&=\left[\Re\left(\m{\rho}_{d}\right)^{T},\Im\left(\m{\rho}_{d}\right)^{T}\right]^{T}\in\mathbb{R}^{2P_{d}}, \\
		\m{\theta}_{\tilde{\m{c}}_{l}}&=\left[\Re\left(\tilde{\m{c}}_{l}\right)^{T},\Im\left(\tilde{\m{c}}_{l}\right)^{T}\right]^{T}\in\mathbb{R}^{2|\mathcal{C}_{l}|}. \label{2.2.6}
	\end{aligned}
\end{equation}
Similarly, the source matrix $\m{S}$ exhibits the coherence structure \eqref{2.2.5} is equivalent to that $\m{\theta}_{\m{S}}$ lies in the constrained set
\begin{equation}
	\begin{aligned}
		\m{\theta}_{\m{S}}\in \mathcal{M}_{2}\triangleq \left\{\m{\theta}_{\m{S}}:\m{\theta}_{\m{S}}=\m{g}_{2}\left(\m{\theta}_{\m{\rho}_{1}},...,\m{\theta}_{\m{\rho}_{D}},\m{\theta}_{\tilde{\m{c}}_{1}},...,\m{\theta}_{\tilde{\m{c}}_{L}}\right)\right\}, \label{2.2.7}
	\end{aligned}
\end{equation}
where $\m{g}_{2}$ is a smooth mapping determined by \eqref{2.2.5}, with its Jacobian matrix, denoted by $\m{J}_{2}$, being of full column rank.
\subsection{Prior Art: CRB Derivation}
The derivation of the CRB under deterministic models has been extensively studied \cite{stoica1989music,stoica1990performance,liu2017cramer,wang2016coarrays,hua1991note}. Under the assumption of Gaussian noise with variance $\sigma$, let $\m{X}=\m{A}(\m{f})\m{S}$ denote the expectation of $\m{Y}$. By denoting the model parameters as $\m{\alpha}$, the FIM of the model parameters that characterizes $\m{X}$, regardless of whether the noise variance $\sigma$ is known or treated as an unknown nuisance parameter, is then given by \cite{stoica1989music}
\begin{equation}
	\begin{aligned}
		\m{I}(\m{\alpha})=\frac{2}{\sigma}\sum_{n=1}^{N}\sum_{l=1}^{L}\Re\left[\left(\frac{\partial X_{nl}}{\partial \m{\alpha}}\right)\left(\frac{\partial X_{nl}}{\partial \m{\alpha}}\right)^{H}\right].\label{2.3.1}
	\end{aligned}
\end{equation}
Denote $\m{b}(f)$ as the derivative of $\m{a}(f)$ with respect to $f$, i.e.,
\begin{equation}
	\begin{aligned}
		\m{b}(f)=\frac{\partial \m{a}(f)}{\partial f}=\left[0,i2\pi e^{i2\pi f},...,i2\pi(N-1)e^{i2\pi(N-1)f}\right]^{T}.\label{2.3.2}
	\end{aligned}
\end{equation}
The partial derivatives involved in (\ref{2.3.1}) are as follows
\begin{equation}
	\begin{aligned}
		\frac{\partial \m{X}_{:,i}}{\partial \m{S}_{:,j}}&=\m{A}(\m{f})\delta_{ij},  \quad
		\frac{\partial \m{X}_{:,i}}{\partial\m{f}}&=\m{B}(\m{f})\m{D}_{i},\label{2.3.4}
	\end{aligned}
\end{equation}
where $\delta_{ij}$ is the Kronecker delta, $\m{B}(\m{f})=[\m{b}(f_{1}),\m{b}(f_{2}),..,$ $\m{b}(f_{K})]$, and $\m{D}_{i}=\textnormal{diag}(\m{S}_{:,i})$. The CRBs for the parameters are given by the corresponding diagonal entries of $\m{I}^{-1}(\m{\alpha})$, denoted by $\m{\Theta}(\m{\alpha})$. For a vector-valued parameter, we use its CRB to denote the trace of the corresponding CRB submatrix. To simplify the notation, we omit the argument $\m{f}$, or $\m{\alpha}$, whenever no ambiguity arises. 

\section{Relationships Between the CRBs}  \label{Relationships}
In this section, we compare the CRBs obtained under the three structural models and examine how different levels of structural information in the source matrix $\m{S}$ affect the estimation bounds for the frequency vector $\m{f}$ and the source matrix $\m{S}$. We first present the main results and then provide their proofs.
\subsection{Main Results}
It is shown in \eqref{2.2.4} and \eqref{2.2.7} that the parameter CRBs under the low-rank and coherence models can be regarded as constrained CRBs obtained by imposing additional structural constraints on the source matrix $\m{S}$.
 The CRB under parametric constraints has been studied in \cite{gorman2002lower,marzetta1993simple,stoica1998cramer}, which can be summarized as the following theorem.
\begin{theorem}\label{Thm 1}
	Let $\m{\alpha}\in\mathbb{R}^{N}$ be the parameter to be estimated with its unconstrained CRB matrix given by $\m{\Theta}$. Consider the parametric constraint $\m{\alpha}\in\{\m{\alpha}:\m{\alpha}=\m{g}(\m{\varphi})\textnormal{ for some }\m{\varphi}\in\mathbb{R}^{M}\}$ with $M<N$, and $\m{g}$ is smooth with its Jacobian matrix $\m{J}=\frac{\partial \m{g}}{\partial \m{\varphi}}\in\mathbb{R}^{N\times M}$ being of full column rank. The constrained CRB matrix of $\m{\alpha}$ is given by
	\begin{equation}
		\begin{aligned}
			\m{\Theta}^{\textnormal{cons}}=\m{\Theta}-\m{\Theta}\m{U}^{T}\left(\m{U}\m{\Theta}\m{U}^{T}\right)^{-1}\m{U}\m{\Theta},\label{3.1}
		\end{aligned}
	\end{equation}
	where $\m{U}\in\mathbb{R}^{(N-M)\times N}$ has orthonormal rows and satisfies $\m{U}\m{J}=\m{0}$.
\end{theorem}
Note that the standard inverse-form CRB requires the FIM to be positive definite, which, under regularity conditions, requires local identifiability of the parameters \cite{rothenberg1971identification}. Therefore, the CRB may still be well defined even when the parameters are not globally identifiable. For example, a classical sufficient and necessary condition for global identifiability is \cite{davies2012rank}
\begin{equation}
	\begin{aligned}
		K<\frac{\textnormal{spark}(\mathcal{A})-1+\textnormal{rank}(\m{X})}{2},\label{3.2}
	\end{aligned}
\end{equation} 
	where $\textnormal{spark}(\mathcal{A})$ is the spark of 
\begin{equation}
	\begin{aligned}
		\mathcal{A}\triangleq\left\{\m{a}(f):f\in[0,1)\right\}, \label{3.3}
	\end{aligned}
\end{equation}
which is definied as the the smallest number of elements in $\mathcal{A}$ that are linearly dependent.
The analysis below also applies to such locally identifiable but globally non-identifiable parameters.\par
In the problem under consideration, the constraints are imposed only on the parameters $\m{\theta}_{\m{S}}$. Therefore, the CRBs for the frequency vector $\m{f}$ and the source matrix $\m{S}$ need to be treated separately. Applying Theorem \ref{Thm 1}, we propose the following proposition.
\begin{proposition}\label{prop2}
	Consider $\m{f}\in\mathbb{R}^{K},\m{\alpha}\in\mathbb{R}^{n}$ as the parameters to be estimated, with its FIM given by
	\begin{equation}
		\begin{aligned}
			\begin{bmatrix}
				\m{I}_{\m{f}\m{f}}&\m{I}_{\m{f}\m{\alpha}}\\
				\m{I}_{\m{\alpha}\m{f}}&\m{I}_{\m{\alpha}\m{\alpha}}
			\end{bmatrix} \label{3.4}
		\end{aligned}
	\end{equation}
	 being positive definite. 
	Let $\m{\Theta}$ denote the corresponding CRB matrix. Suppose that $\m{\alpha}$ is further constrained as $\m{\alpha}\in\{\m{\alpha}:\m{\alpha}=\m{g}(\m{\varphi})\textnormal{ for some }\m{\varphi}\in\mathbb{R}^{m}\}$ with $m<n$, and $\m{g}$ is smooth with its Jacobian matrix $\m{J}=\frac{\partial \m{g}}{\partial \m{\varphi}}\in\mathbb{R}^{n\times m}$ being of full column rank. Let $\m{\Theta}^{\textnormal{cons}}$ denote the resulting constrained CRB. Partitioning both $\m{\Theta}$ and $\m{\Theta}^{\textnormal{cons}}$ in the same manner as in \eqref{3.4}, we have
	\begin{equation}
		\begin{aligned}
			\m{\Theta}_{\m{\alpha\alpha}}\succeq \m{\Theta}_{\m{\alpha\alpha}}^{\textnormal{cons}},\quad\textnormal{tr}\left(\m{\Theta}_{\m{\alpha\alpha}}- \m{\Theta}_{\m{\alpha\alpha}}^{\textnormal{cons}}\right)>0, \label{3.6}
		\end{aligned}
	\end{equation}
	and
	\begin{equation}
		\begin{aligned}
			\m{\Theta}_{\m{ff}}\succeq \m{\Theta}_{\m{ff}}^{\textnormal{cons}}, \label{3.7}
		\end{aligned}
	\end{equation}
	with equality of \eqref{3.7} if and only if
	\begin{equation}
		\begin{aligned}
			\textnormal{col}(\m{I}_{\m{\alpha}\m{\alpha}}^{-1}\m{I}_{\m{\alpha}\m{f}})\subseteq \textnormal{col}(\m{J}). \label{3.8}
		\end{aligned}
	\end{equation}
\end{proposition}
\begin{proof}
	See Appendix \ref{Appendix A}.
\end{proof}
The condition \eqref{3.8} admits a geometric interpretation. The columns
of $\m{I}_{\m{\alpha}\m{\alpha}}^{-1}\m{I}_{\m{\alpha}\m{f}}$ are nuisance-space directions induced by local changes in the parameters of interest. A nuisance constraint
leaves the frequency CRB unchanged exactly when all these
directions lie in the tangent space $\m{J}$ of the constrained nuisance parameter manifold. Proposition \ref{prop2} indicates that imposing additional constraints on $\m{S}$ reduces the CRB for $\m{S}$. Whether a strict reduction in the CRB for $\m{f}$ can be achieved is determined by the condition in \eqref{3.8}. Building on this result, we obtain the following theorems.
\begin{theorem}[Rank Information]\label{Thm 3}
	Denote $\m{\Theta}^{\textnormal{LR}}$ as the constrained CRB matrix for the parameters $\m{f},\m{\theta}_{\m{S}}$ under the constraint \eqref{2.2.4} and partition it in the same manner as in \eqref{3.4}. We have
	\begin{equation}
		\begin{aligned}
	        \m{\Theta}_{\m{f}\m{f}}=\m{\Theta}_{\m{f}\m{f}}^{\textnormal{LR}}. \label{3.9}
		\end{aligned}
	\end{equation} 
	When $R<\min\{K,L\}$, we have $\textnormal{tr}\left(\m{\Theta}_{\m{\theta}_{\m{S}}\m{\theta}_{\m{S}}}- \m{\Theta}_{\m{\theta}_{\m{S}}\m{\theta}_{\m{S}}}^{\textnormal{LR}}\right)>0$.
\end{theorem}
\begin{theorem}[Coherence Structure]\label{Thm 4}
	Suppose that $D\geq2$ and there exist $d\in\{1,2,...,D\}$ such that $P_{d}>1$. Denote $\m{\Theta}^{\textnormal{Cor}}$ as the constrained CRB matrix for the parameters $\m{f},\m{\theta}_{\m{S}}$ under the constraint \eqref{2.2.7} and partition it in the same manner as in \eqref{3.4}. We have 
	$	\m{\Theta}_{\m{\theta}_{\m{S}}\m{\theta}_{\m{S}}}\succeq \m{\Theta}_{\m{\theta}_{\m{S}}\m{\theta}_{\m{S}}}^{\textnormal{Cor}},\quad \textnormal{tr}\left(\m{\Theta}_{\m{\theta}_{\m{S}}\m{\theta}_{\m{S}}}- \m{\Theta}_{\m{\theta}_{\m{S}}\m{\theta}_{\m{S}}}^{\textnormal{Cor}}\right)>0$ for any pairwise-incoherent $\m{C}$. 
	Moreover,  for generic\footnote{Here, ``generic" means excluding a set of Lebesgue measure zero.} $\left(\m{f},\m{\rho}_{1},...,\m{\rho}_{D}\right)$, we have
	\begin{equation}
		\begin{aligned}
			\m{\Theta}_{\m{f}\m{f}}\succeq \m{\Theta}_{\m{f}\m{f}}^{\textnormal{Cor}},\quad \textnormal{tr}\left(\m{\Theta}_{\m{f}\m{f}}- \m{\Theta}_{\m{f}\m{f}}^{\textnormal{Cor}}\right)>0. \label{3.11}
		\end{aligned}
	\end{equation} 
\end{theorem}
The assumptions in Theorem \ref{Thm 4} are imposed to exclude degenerate cases. When $D=1$, the coherent model degenerates to the low-rank model. When $P_d=1$ for all $d$, no nontrivial multipath-induced coherence exists, and the coherent model degenerates to the general model. Therefore, these assumptions ensure that the coherent structure is nontrivial. \par
	
Theorem \ref{Thm 3} shows that the low-rank constraint preserves all frequency-coupled nuisance directions, leaving the frequency CRB unchanged. In contrast, Theorem \ref{Thm 4} shows that the coherence constraint removes some of these directions, thereby improving the frequency CRB. Overall, both low-rank and coherent structures can reduce the CRB for $\m{S}$; however, the CRB for $\m{f}$ is unaffected by the low-rank structure alone and can be strictly reduced by further exploiting the coherence structure.
\begin{remark}
		It is noted that, although this paper focuses on the coherence model, Theorem \ref{Thm 3} relies only on the low-rank structure of the source matrix. Therefore, it is not restricted to the coherent setting considered here, but are applicable to general low-rank models.
\end{remark}
\subsection{Proof of Theorem \ref{Thm 3}}  \label{proof of 3}
We mainly consider the case $R<\min\{K,L\}$, i.e., the low-rank constraint is nontrivial. When $R=\min\{K,L\}$, the rank constraint is trivial and the low-rank model coincides with the general model; hence the two CRB matrices are identical. The proof of Theorems \ref{Thm 3} is established by showing that the condition \eqref{3.8} in Proposition \ref{prop2} is satisfied. To this end, we first derive the FIM blocks $\m{I}_{\m{\theta}_{\m{S}}\m{\theta}_{\m{S}}}$ and $\m{I}_{\m{\theta}_{\m{S}}\m{f}}$. It is obtained from (\ref{2.3.4}) that $\m{I}_{\m{\theta}_{\m{S}}\m{\theta}_{\m{S}}}=\mathcal{R}(\m{H})$ and $\m{I}_{\m{\theta}_{\m{S}}\m{f}}=\mathcal{V}(\m{L})$,
with
\begin{equation}
	\begin{aligned}
		\m{H}=\begin{bmatrix}
			\m{A}^{H}\m{A}& & & \\
			&\m{A}^{H}\m{A}& &\\
			& & \ddots &\\
			& & & \m{A}^{H}\m{A}
		\end{bmatrix}\in\mathbb{C}^{KL\times KL}, \label{3.12}
	\end{aligned}
\end{equation}
and
\begin{equation}
	\begin{aligned}
		\m{L}=\begin{bmatrix}
			\m{A}^{H}\m{B}\m{D}_{1}\\
			\m{A}^{H}\m{B}\m{D}_{2}\\
			\vdots\\
			\m{A}^{H}\m{B}\m{D}_{L}
		\end{bmatrix}\in\mathbb{C}^{KL\times K}, \label{3.13}
	\end{aligned}
\end{equation}
where where the common scaling factor $\frac{2}{\sigma}$ is omitted for notational simplicity.
We now compute the Jacobian matrix $\m{J}_{1}$, for which the required partial derivatives are given below 
\begin{subequations} \label{3.14}
	\begin{align}
		\frac{\partial \m{S}_{:,i}}{\partial \m{\Phi}_{:,j}}&=\m{I}_{K}\delta_{ij}, &\quad&i=1,2,...,R, \label{3.14a}\\
		\frac{\partial \m{S}_{:,i}}{\partial \m{\Phi}_{:,j}}&=\tilde{M}_{j,i-R}\m{I}_{K}\delta_{ij}, &\quad &i=R+1,...,L, \label{3.14b}\\
		\frac{\partial \m{S}_{:,i}}{\partial \tilde{\m{M}}_{:,j}}&=\m{0}_{K\times R},  &\quad&i=1,2,...,R, \label{3.14c} \\
		\frac{\partial \m{S}_{:,i}}{\partial \tilde{\m{M}}_{:,j}}&=\m{\Phi}\delta_{i-R,j},  &\quad&i=R+1,...,L. \label{3.14d}
	\end{align}
\end{subequations} 
Summarizing \eqref{3.14}, we get $\m{J}_{1}$ admits the form of $\m{J}_{1}=\mathcal{R}(\m{K}_{1})$
with
\begin{equation}
	\begin{aligned}
		\m{K}_{1}=\begin{bmatrix}
			\m{I}_{KR}&\m{0}_{KR\times R(L-R)}\\
			\tilde{\m{M}}^{T}\otimes\m{I}_{K}&\m{I}_{L-R}\otimes\m{\Phi}
		\end{bmatrix}. \label{3.15}
	\end{aligned}
\end{equation}
To facilitate the analysis, we present the following lemmas.
\begin{lemma}\label{Lem5}
	For matrices $\m{A}\in\mathbb{C}^{m\times n},\m{B}\in\mathbb{C}^{n\times l}$, we have
	\begin{subequations}\label{3.16}
		\begin{align}
			\mathcal{R}(\m{A}\m{B})&=\mathcal{R}(\m{A})\cdot\mathcal{R}(\m{B}).\label{3.16a}\\
			\mathcal{V}(\m{A}\m{B})&=\mathcal{R}(\m{A})\cdot\mathcal{V}(\m{B}). \label{3.16b}
		\end{align}
	\end{subequations}
\end{lemma}
\begin{proof}
	See Appendix \ref{Appendix B}.
\end{proof}
\begin{lemma}\label{Lem6}
	For an invertible matrix $\m{A}\in\mathbb{C}^{m\times m}$, we have
	\begin{equation}
		\begin{aligned}
			\mathcal{R}(\m{A})^{-1}=\mathcal{R}(\m{A}^{-1}).\label{3.17}
		\end{aligned}
	\end{equation}
\end{lemma}
\begin{proof}
	The result follows directly from Lemma \ref{Lem5}.
\end{proof}
\begin{lemma}\label{Lem7}
	For $\m{A}\in\mathbb{C}^{m\times n}$ and $\m{B}\in\mathbb{C}^{m\times l}$, $\textnormal{col}(\m{A})\subseteq\textnormal{col}(\m{B})$ if and only if $\textnormal{col}(\mathcal{V}(\m{A}))\subseteq\textnormal{col}(\mathcal{R}(\m{B}))$.
\end{lemma}
\begin{proof}
	See Appendix \ref{Appendix C}.
\end{proof}
To prove Theorem \ref{Thm 3}, it suffices to show that there exists a matrix $\m{T}\in\mathbb{C}^{(KR+RL-R^{2})\times K}$ such that  
\begin{equation}
	\begin{aligned}
		\m{H}^{-1}\m{L}=\m{K}_{1}\m{T}.\label{3.20}
	\end{aligned}
\end{equation}
Let $\m{P}=\left(\m{A}^{H}\m{A}\right)^{-1}\m{A}^{H}\m{B}$ and partition $\m{T}$ as 
\begin{equation}
	\begin{aligned}
		\m{T}=\begin{bmatrix}
			\m{T}_{1}^{T}&
			\cdots&
			\m{T}_{R}^{T}&
			\m{T}_{R+1}^{T}
		\end{bmatrix}^{T}, \label{3.21} 
	\end{aligned}
\end{equation}
where $\m{T}_{r}\in\mathbb{C}^{K\times K}$ and $\m{T}_{R+1}$ has the remaining compatible dimension. We now show that 
\begin{equation}
	\begin{aligned}
		\m{T}_{r}=\m{P}\m{D}_{r},\quad r=1,2,...,R, \quad
		\m{T}_{R+1}=\m{0} \label{3.22}
	\end{aligned}
\end{equation}
is exactly the solution to (\ref{3.20}).
By performing block matrix multiplication and setting $\m{T}_{R+1}=\m{0}$, (\ref{3.20}) is equivalent to
\begin{subequations}
	\begin{align}
		\m{P}\m{D}_{r}&=\m{T}_{r} &\quad &r=1,2,...,R, \label{3.23a}\\
		\m{P}\m{D}_{r}&=\left(\tilde{\m{M}}_{:,r-R}^{T}\otimes \m{I}_{K}\right)\begin{bmatrix}
			\m{T}_{1}\\
			\vdots\\
			\m{T}_{R}
		\end{bmatrix}&\quad &r=R+1,...,L. \label{3.23b}
	\end{align}
\end{subequations}
Following the factorization (\ref{2.2.2}) and the definition of $\tilde{\m{M}}$, the right side of (\ref{3.23b}) equals
\begin{equation}
	\begin{aligned}
		\sum_{i=1}^{R}\tilde{M}_{i,r-R}\m{T}_{r}&=\sum_{i=1}^{R}\tilde{M}_{i,r-R}\m{P}\m{D}_{r}\quad (\textnormal{based on (\ref{3.23a})})\\
		&=\m{P}\sum_{i=1}^{R}\tilde{M}_{i,r-R}\m{D}_{r}\\
		&=\m{P}\cdot\textnormal{diag}\left(\sum_{i=1}^{R}\tilde{M}_{i,r-R}\m{S}_{:,i}\right)\\
		&=\m{P}\cdot\textnormal{diag}(\m{S}_{:,r})\quad (\textnormal{based on (\ref{2.2.2})})\\
		&=\m{P}\m{D}_{r}.\label{3.24}
	\end{aligned}
\end{equation}
Therefore, (\ref{3.23a}) and (\ref{3.23b}) hold for the $\m{T}$ defined in (\ref{3.22}), completing the proof.
\subsection{Proof of Theorem \ref{Thm 4}} \label{proof of 4}
To calculate $\m{J}_{2}$, we define $\m{\Gamma}_{l}=\left[\m{e}_{d}\right]_{d\in\mathcal{C}_{l}}\in\mathbb{R}^{D\times |\mathcal{C}_{l}|}$, where $\m{e}_{d}$ denotes the $d$-th canonical basis vector in $\mathbb{R}^{D}$. Then we obtain that $\m{\Gamma}_{l}^{T}\m{C}_{:,l}=\tilde{\m{c}}_{l}$. The partial derivatives required in $\m{J}_{2}$ are given below
\begin{subequations}\label{3.25}
	\begin{align}
		\frac{\partial \m{S}_{:,i}}{\partial \m{\rho}} &= \m{\Xi}_{i}, &\quad & i=1,2,...,L \label{3.25b}\\
		\frac{\partial\m{S}_{:,i}}{\partial \tilde{\m{c}}_{j}} & = \m{\Psi}\m{\Gamma}_{j}\delta_{i,j}, &\quad &\forall i,j,\label{3.25c}
	\end{align}
\end{subequations}
where $\m{\rho}=\left[\m{\rho}_{1}^{T},\m{\rho}_{2}^{T},...,\m{\rho}_{D}^{T}\right]^{T}$ and 
\begin{equation}
	\begin{aligned}
		\m{\Xi}_{i}=\textnormal{blkdiag}(C_{1i}\m{I}_{P_{1}},C_{2i}\m{I}_{P_{2}},...,C_{Di}\m{I}_{P_{D}})\label{3.23}
	\end{aligned}
\end{equation}
is a diagonal matrix. Based on \eqref{3.25}, $\m{J}_{2}$ admits the form of $\m{J}_{2}=\mathcal{R}(\m{K}_{2})$
with
\begin{equation}
	\begin{aligned}
		\m{K}_{2}=\begin{bmatrix}
			\m{\Xi}_{1}&\m{\Psi}\m{\Gamma}_{1}&\m{0}_{K\times |\mathcal{C}_{2}|}&\cdots&\m{0}_{K\times |\mathcal{C}_{L}|}\\
			\m{\Xi}_{2}&\m{0}_{K\times |\mathcal{C}_{1}|}&\m{\Psi}\m{\Gamma}_{2}&\cdots&\m{0}_{K\times |\mathcal{C}_{L}|}\\
			\vdots&\vdots&\vdots&\ddots&\vdots\\
			\m{\Xi}_{L}&\m{0}_{K\times |\mathcal{C}_{1}|}&\m{0}_{K\times |\mathcal{C}_{2}|}&\cdots&\m{\Psi}\m{\Gamma}_{L}
		\end{bmatrix}.\label{3.26}
	\end{aligned}
\end{equation}
We prove Theorem \ref{Thm 4} by contradiction. Suppose that there exists a matrix $\m{T}\in\mathbb{C}^{(K+DL-D)\times K}$ such that
\begin{equation}
	\begin{aligned}
		\m{H}^{-1}\m{L}=\m{K}_{2}\m{T}.\label{3.27}
	\end{aligned}
\end{equation}
We partition $\m{T}$ as 
\begin{equation}
	\begin{aligned}
		\m{T}=\begin{bmatrix}
			\m{T}_{\m{\rho}}^{T}&
			\tilde{\m{T}}_{1}^{T}&
			\tilde{\m{T}}_{2}^{T}&
			\cdots&
			\tilde{\m{T}}_{L}^{T}
		\end{bmatrix}^{T},\label{3.28}
	\end{aligned}
\end{equation}
where $\m{T}_{\m{\rho}}\in\mathbb{C}^{K\times K}$, and $\tilde{\m{T}}_{l}\in\mathbb{C}^{|\mathcal{C}_{l}|\times K}$ for $l=1,2,...,L$. Let $\m{T}_{l}=\m{\Gamma}_{l}\tilde{\m{T}}_{l}\in\mathbb{C}^{D\times K}$. Then (\ref{3.27}) is equivalent to 
\begin{equation}
	\begin{aligned}
		\m{P}\m{D}_{l}=\m{\Xi}_{l}\m{T}_{\m{\rho}}+\m{\Psi}\m{T}_{l},\quad l=1,2,...,L.  \label{3.29}
	\end{aligned}
\end{equation}
We partition the index set $\{1,2,...,K\}$ into $D$ consecutive subsets $\mathcal{I}_{1},\mathcal{I}_{2},...,\mathcal{I}_{D}$, where $|\mathcal{I}_{d}|=P_{d}$, and partition $\m{P}$ accordingly as 
\begin{equation}
	\begin{aligned}
		\begin{bmatrix}
			\m{P}_{11}&\m{P}_{12}&\cdots&\m{P}_{1D}\\
			\m{P}_{21}&\m{P}_{22}&\cdots&\m{P}_{2D}\\
			\vdots&\vdots&\ddots&\vdots\\
			\m{P}_{D1}&\m{P}_{D2}&\cdots&\m{P}_{DD}
		\end{bmatrix}, \label{3.30}
	\end{aligned}
\end{equation}
where $\m{P}_{ij}\in\mathbb{C}^{P_{i}\times P_{j}}$. Suppose that $P_{i}>1$, and choose any $j\neq i$. Taking the $(i,j)$-th block of \eqref{3.29} gives that
\begin{equation}
	\begin{aligned}
		C_{jl}\m{P}_{ij}\textnormal{diag}(\m{\rho}_{j})=C_{il}\left(\m{T}_{\m{\rho}}\right)_{\mathcal{I}_{i},\mathcal{I}_{j}}+\m{\rho}_{i}\left(\m{T}_{l}\right)_{i,\mathcal{I}_{j}}, \label{3.31}
	\end{aligned}
\end{equation}
for $l=1,2,...,L$. Multiplying both sides by $\m{I}_{P_{i}}-\m{\rho}_{i}\m{\rho}_{i}^{\dagger}$ yields that 
\begin{equation}
	\begin{aligned}
		C_{jl}\left(\m{I}_{P_{i}}-\m{\rho}_{i}\m{\rho}_{i}^{\dagger}\right)\m{P}_{ij}\textnormal{diag}(\m{\rho}_{j})=C_{il}\left(\m{I}_{P_{i}}-\m{\rho}_{i}\m{\rho}_{i}^{\dagger}\right)\left(\m{T}_{\m{\rho}}\right)_{\mathcal{I}_{i},\mathcal{I}_{j}}. \label{3.32}
	\end{aligned}
\end{equation}
Since the $i$-th row and the $j$-th row of $\m{C}$ are incoherent, there exist $l_{1}$ and $l_{2}$, such that $C_{il_{1}}C_{jl_{2}}-C_{il_{2}}C_{jl_{1}}\neq0.$ Applying \eqref{3.32} for $l=l_{1}$ and $l=l_{2}$, we obtain 
\begin{equation}
	\begin{aligned}
		\begin{bmatrix}
			C_{jl_{1}}\m{I}_{P_{i}} & -C_{il_{1}}\m{I}_{P_{i}}\\
			C_{jl_{2}}\m{I}_{P_{i}} & -C_{il_{2}}\m{I}_{P_{i}}
		\end{bmatrix} \begin{bmatrix}
		\left(\m{I}_{P_{i}}-\m{\rho}_{i}\m{\rho}_{i}^{\dagger}\right)\m{P}_{ij}\textnormal{diag}(\m{\rho}_{j})\\
		\left(\m{I}_{P_{i}}-\m{\rho}_{i}\m{\rho}_{i}^{\dagger}\right)\left(\m{T}_{\m{\rho}}\right)_{\mathcal{I}_{i},\mathcal{I}_{j}}
		\end{bmatrix}=\begin{bmatrix}
		\m{0}\\
		\m{0}
		\end{bmatrix}. \label{3.33}
	\end{aligned}
\end{equation} 
The coefficient matrix in \eqref{3.33} is non-singular since $C_{il_{1}}C_{jl_{2}}-C_{il_{2}}C_{jl_{1}}\neq0$. Hence we have
\begin{equation}
	\begin{aligned}
		\left(\m{I}_{P_{i}}-\m{\rho}_{i}\m{\rho}_{i}^{\dagger}\right)\m{P}_{ij}=\m{0}_{P_{i},P_{j}}.\label{3.35}
	\end{aligned}
\end{equation}
We prove Theorem \ref{Thm 4} by showing that \eqref{3.35} does not hold for generic parameter values. It can be calculated that
\begin{subequations}\label{3.36}
	\begin{align}
		\left[\m{A}^{H}\m{A}\right]_{ij}&=\sum_{n=1}^{N}e^{i2\pi(n-1)(f_{j}-f_{i})},\label{3.36a}\\
		\left[\m{A}^{H}\m{B}\right]_{ij}&=i2\pi\sum_{n=1}^{N}(n-1)e^{i2\pi(n-1)(f_{j}-f_{i})},\label{3.36b}
	\end{align}
\end{subequations}
where are both analytic functions of $\m{f}$. Then we get that each element of $\m{P}$ is analytic for $\m{f}$. In addition, the $(k,l)$-element of the matrix $\m{\rho}_{i}\m{\rho}_{i}^{\dagger}$ is
\begin{equation}
	\begin{aligned}
		\frac{\left(\m{\rho}_{i}\right)_{k}\overline{\left(\m{\rho}_{i}\right)_{l}}}{\left|\left(\m{\rho}_{i}\right)_{1}\right|^{2}+\left|\left(\m{\rho}_{i}\right)_{2}\right|^{2}+\cdots+\left|\left(\m{\rho}_{i}\right)_{P_{i}}\right|^{2}},\label{3.37}
	\end{aligned}
\end{equation}
whose real and imaginary parts are analytic for $\m{\rho}$. Therefore, we conclude that the real and imaginary parts of each element of  $\left(\m{I}_{P_{i}}-\m{\rho}_{i}\m{\rho}_{i}^{\dagger}\right)\m{P}_{ij}$ are analytic for $\m{f}$ and $\m{\rho}$.
Since a nontrivial analytic function vanishes only on a set of Lebesgue measure zero, it suffices to construct $\m{f}$ and $\m{\rho}$ such that $\left(\m{I}_{P_{i}}-\m{\rho}_{i}\m{\rho}_{i}^{\dagger}\right)\m{P}_{ij}$ is not a zero matrix. Indeed, whenever $\m{P}_{ij}\neq \m{0}$, one can choose $\m{\rho}_{i}$ such that it does not lie in the one-dimensional subspace spanned by any nonzero column of $\m{P}_{ij}$. Hence, it suffices to construct $\m{f}$ such that $\m{P}_{ij}\neq \m{0}$, which is established in Lemma \ref{Lem8}. This completes the proof of Theorem \ref{Thm 4}.
\begin{lemma}\label{Lem8}
	For generic $\m{f}$, the matrix $\m{P}_{ij}$ is nonzero fro any $i,j$.
\end{lemma}
\begin{proof}
	See Appendix \ref{Appendix D}.
\end{proof}

	It is worth noting that the genericity requirement in Theorem 3 is essential, since there exist exceptional parameter values for which the two frequency CRBs coincide. Take $l=l_{i}$ in the equation \eqref{3.31}. By the definition of $\m{\Gamma}_{l}$ and $\m{T}_{l}=\m{\Gamma}_{l}\tilde{\m{T}}_{l}$, we have $\left(\m{T}_{l}\right)_{i,:}=\m{0}$ and $C_{il_{i}}=1$, thus
	\begin{equation}
		\begin{aligned}
			C_{jl_{i}}\m{P}_{ij}\textnormal{diag}(\m{\rho}_{j})=\left(\m{T}_{\m{\rho}}\right)_{\mathcal{I}_{i},\mathcal{I}_{j}}. 
		\end{aligned}
	\end{equation}
	Substituting this expression back yields that
	\begin{equation}
		\begin{aligned}
			\left(C_{jl}-C_{il}C_{jl_{i}}\right)\m{P}_{ij}\textnormal{diag}(\m{\rho}_{j})=\m{\rho}_{i}\left(\m{T}_{l}\right)_{i,\mathcal{I}_{j}}. \label{sufficient}
		\end{aligned}
	\end{equation}
	Therefore, for parameter values in the set
	\begin{equation}
		\begin{aligned}
			\{
			(\m{f},\m{C},\m{\rho}):
			\textnormal{col}(\m{P}_{ij})\subseteq \textnormal{col}&(\m{\rho}_i),
			\	\forall \ i\neq j
			\}, \label{3.38}
		\end{aligned}
	\end{equation}
	equation \eqref{3.29} admits a solution, since the left-hand side of \eqref{sufficient} always lies in $\textnormal{col}(\m{\rho}_{i})$, and the corresponding $\m{T}_{l}$ can be found. As a concrete example, consider the case of $L=2,K=3,D=2$ with $P_{1}=2,P_{2}=1$. For $\m{f}$ such that $\m{P}_{12}\neq\m{0}$, let 
	\begin{equation}
		\begin{aligned}
			\m{\rho}_{1}=\m{P}_{12},\quad\rho_{2}=1,\label{3.39}
		\end{aligned}
	\end{equation}
	then \eqref{3.38} holds for any pairwise-incoherent $\m{C}$, and hence $\m{\Theta}_{\m{f}\m{f}}=\m{\Theta}^{\textnormal{Cor}}_{\m{f}\m{f}}$.

\section{Extensions}  \label{extensions}
In this section, we first generalize the preceding results to general array geometries while retaining the one-dimensional DOA setting, and then extend them further to multidimensional DOA estimation.
\subsection{General Array Geometries}
For general array geometries, such as sparse linear arrays \cite{moffet1968minimum,pal2010nested,vaidyanathan2010sparse} and general nonuniform linear arrays with arbitrary sensor locations, the steering vector does not admit the form given in \eqref{2.1.2}. In this setting, we require to re-examine whether the preceding conclusions still hold. Since the proof of Theorem \ref{Thm 3} does not rely on the specific form of $\m{a}(f)$, Theorem \ref{Thm 3} remains valid for general array geometries. For Theorem \ref{Thm 4}, assuming that $\m{a}(f)$ is analytic, it suffices to construct an $\m{f}$ such that $\m{P}_{ij}\neq \m{0}$. To this end, we note that 
\begin{equation}
	\begin{aligned}
		\m{P}_{ij}=\left[\left(\m{A}^{H}\m{A}\right)^{-1}\m{A}^{H}\right]_{\mathcal{I}_{i},:}\m{B}_{:,\mathcal{I}_{j}}.\label{5.1.1}
	\end{aligned}
\end{equation}
Letting $\m{A}_{\mathcal{I}_{i}}=\m{A}_{:,\mathcal{I}_{i}}$, $\m{B}_{\mathcal{I}_{i}}=\m{B}_{:,\mathcal{I}_{i}}$, and $\m{A}_{-\mathcal{I}_{i}}$
denote the submatrix of $\m{A}$ obtained after excluding $\m{A}_{\mathcal{I}_i}$, we provide an alternative representation of $\m{P}_{ij}$ in the following lemma.
\begin{lemma}\label{lemma 9}
	The matrix $\m{P}_{ij}$ admits an alternative representation as 
	\begin{equation}
		\begin{aligned}
			\m{P}_{ij}=\left(\m{A}_{\mathcal{I}_{i}}^{H}\m{\Pi}_{\m{A}_{-\mathcal{I}_{i}}}^{\perp}\m{A}_{\mathcal{I}_{i}}\right)^{-1}\m{A}_{\mathcal{I}_{i}}^{H}\m{\Pi}_{\m{A}_{-\mathcal{I}_{i}}}^{\perp}\m{B}_{\mathcal{I}_{j}},\label{5.1.2}
		\end{aligned}
	\end{equation}
	where $\m{\Pi}_{\m{A}_{-\mathcal{I}_{i}}}^{\perp}$ denotes the projection onto the orthogonal complement of the column space of $\m{A}_{-\mathcal{I}_i}$.
\end{lemma}
\begin{proof}
	See Appendix \ref{Appendix E}.
\end{proof}
By Lemma \ref{lemma 9}, we show that under mild conditions on $\m{a}(f)$, it is still possible to construct an $\m{f}$ such that $\m{P}_{ij}\neq \m{0}$.
\begin{lemma}\label{lemma 10}
	Suppose that the steering vector $\m{a}(f)$ is analytic and satisfies the following conditions: \textit{(1)} $\m{a}(f)$ spans $\mathbb{C}^N$, and \textit{(2)} there exists an $f$ such that $\m{b}(f)$ is not proportional to $\m{a}(f)$. For $P_{i}>1$ and $i\neq j$, one can construct $\m{f}$ such that $\m{P}_{ij}\neq \m{0}$. 
\end{lemma}
\begin{proof}
	See Appendix \ref{Appendix F}.
\end{proof}
The conditions required by Lemma \ref{lemma 10} are satisfied by most commonly used array geometries \cite{krim2002two}. Hence, Theorem \ref{Thm 4} remains valid for a broad class of array geometries beyond the specific array model considered earlier.
\subsection{Multidimensional DOA Estimation}
In this section, we consider multidimensional DOA estimation, which arises in a wide range of array geometries, including uniform planar arrays (UPAs) \cite{zoltowski1996closed} and uniform circular arrays (UCAs) \cite{mathews1994eigenstructure}. For clarity, we focus on the two-dimensional case, while the analysis can be extended to higher-dimensional models arising in related problems, such as multidimensional line spectral estimation \cite{yang2016vandermonde}. Specifically, the two-dimensional DOA model is formulated as follows. 
\begin{equation}
	\begin{aligned}
		\m{Y}=\m{A}(\m{F})\m{S}+\m{E}, \label{5.2.1}
	\end{aligned}
\end{equation}
where $\m{F}=[\m{f}^{1},\m{f}^{2}]\in\mathbb{R}^{K\times 2}$, and
\begin{equation}
	\begin{aligned}
		\m{A}(\m{F})=\left[\m{a}(f_{1}^{1},f_{1}^{2}),\m{a}(f_{2}^{1},f_{2}^{2}),...,\m{a}(f_{K}^{1},f_{K}^{2})\right] \label{5.2.2}
	\end{aligned}
\end{equation} 
with $\m{a}(f^{1},f^{2})$ denoting the steering vector associated with the two-dimensional frequency pair $(f^{1},f^{2})$. In this case, the parameter vector is modified as $\left[\m{\theta}_{\m{F}}^{T},\m{\theta}_{\m{S}}^{T}\right]^{T}\in\mathbb{R}^{2K+2KL}$, where $\m{\theta}_{\m{F}}=\textnormal{vec}(\m{F})$ collects the multidimensional frequency parameters. Accordingly, the FIM calculation differs from that in the one-dimensional case. Let
\begin{equation}
	\begin{aligned}
		\m{b}^{1}(\m{f})=\frac{\partial \m{a}(\m{f})}{\partial f^{1}},\quad \m{b}^{2}(\m{f})=\frac{\partial \m{a}(\m{f})}{\partial f^{2}}. \label{5.2.3}
	\end{aligned}
\end{equation}
The partial derivatives required to calculate FIM is given by
\begin{equation}
	\begin{aligned}
		\frac{\partial \m{X}_{:,i}}{\partial \m{f}^{1}}=\m{B}^{1}(\m{F})\m{D}_{i},\quad 
		\frac{\partial \m{X}_{:,i}}{\partial \m{f}^{2}}=\m{B}^{2}(\m{F})\m{D}_{i}, \label{5.2.4}
	\end{aligned}
\end{equation}
where $\m{B}^{1}(\m{F})=\left[\m{b}^{1}(\m{f}_{1}),\m{b}^{1}(\m{f}_{2}),...,\m{b}^{1}(\m{f}_{K})\right]$ and $\m{B}^{2}(\m{F})$ is defined similarly. Let $\m{B}=\left[\m{B}^{1},\m{B}^{2}\right]$. Due to the change from $\m{f}$ to $\m{F}$ in the parameterization, the matrices involved in the above proofs in Section \ref{proof of 3} and Section \ref{proof of 4} remain essentially the same, except that $\m{I}_{\m{\theta}_{\m{S}}\m{f}}$ is replaced by $\m{I}_{\m{\theta}_{\m{S}}\m{\theta}_{\m{F}}}=\mathcal{V}(\m{L})$ with
\begin{equation}
	\begin{aligned}
		\m{L} = \begin{bmatrix}
			\m{A}^{H}\m{B}\left(\m{I}_{2}\otimes \m{D}_{1}\right)\\
			\m{A}^{H}\m{B}\left(\m{I}_{2}\otimes \m{D}_{2}\right)\\
			\vdots\\
			\m{A}^{H}\m{B}\left(\m{I}_{2}\otimes \m{D}_{L}\right)
		\end{bmatrix}. \label{5.2.5}
	\end{aligned}
\end{equation}
By the linearity of the Kronecker product, we have that
\begin{equation}
	\begin{aligned}
		\m{T}_{r}=\m{P}\left(\m{I}_{2}\otimes\m{D}_{r}\right),\quad r=1,2,...,R, \quad
		\m{T}_{R+1}=\m{0} \label{5.2.6}
	\end{aligned}
\end{equation}
is the solution to \eqref{3.20}, where 
\begin{equation}
	\begin{aligned}
		\m{P}&=\left(\m{A}^{H}\m{A}\right)^{-1}\m{A}^{H}\m{B}
		\\&=\left[\left(\m{A}^{H}\m{A}\right)^{-1}\m{A}^{H}\m{B}^{1},\left(\m{A}^{H}\m{A}\right)^{-1}\m{A}^{H}\m{B}^{2}\right]\in\mathbb{C}^{K\times 2K}.\label{5.2.7}
	\end{aligned}
\end{equation}
Thus we conclude that Theorem \ref{Thm 3} can be extended to the multidimensional case. \par 
To extend Theorem \ref{Thm 4} to the two-dimensional case, we partition $\m{P}$ as $\m{P}=[\m{P}^{1},\m{P}^{2}]$ with $\m{P}^{i}\in\mathbb{C}^{K\times K}$. Following similar steps in Theorem \ref{Thm 4}, we obtain an expression analogous to \eqref{3.29}:
\begin{equation}
	\begin{aligned}
		\m{P}\left(\m{I}_{2}\otimes\m{D}_{l}\right)=\m{\Xi}_{l}\m{T}_{\m{\rho}}+\m{\Psi}\m{T}_{l},\quad l=1,2,...,L, \label{5.2.8}
	\end{aligned}
\end{equation}
where $\m{T}_{\m{\rho}}\in\mathbb{C}^{K\times 2K},\m{T}_{l}\in\mathbb{C}^{|\mathcal{C}_{l}|\times 2K}$ require to be constructed. Partition $\m{T}_{\m{\rho}}$ and $\m{T}_{l}$ as $\m{T}_{\m{\rho}}=\left[\m{T}_{\m{\rho}}^{1},\m{T}_{\m{\rho}}^{2}\right]$ and $\m{T}_{l}=\left[\m{T}_{l}^{1},\m{T}_{l}^{2}\right]$, respectively. Equation \eqref{5.2.8} is equivalent to the following two subequations as
\begin{subequations}
    \begin{align}
    	\m{P}^{1}\m{D}_{l}&=\m{\Xi}_{l}\m{T}_{\m{\rho}}^{1}+\m{\Psi}\m{T}_{l}^{1},\\
    	\m{P}^{2}\m{D}_{l}&=\m{\Xi}_{l}\m{T}_{\m{\rho}}^{2}+\m{\Psi}\m{T}_{l}^{2}.
    \end{align}
\end{subequations}
Each equation has exactly the same form as its one-dimensional counterpart. Hence, the argument used in the proof of Theorem \ref{Thm 4} applies verbatim to each equation after replacing $\m{P}$ with $\m{P}^{i},i=1,2$.
Partition $\m{P}^{1}$ and $\m{P}^{2}$ in the same manner as in \eqref{3.30}. Based on a similar argument, it suffices to construct $\m{F}$ such that $\m{P}^{1}_{ij}\neq\m{0}$ or $\m{P}^{2}_{ij}\neq\m{0}$ for some $P_{i}>1$ and $i\neq j$, which holds if $\m{a}(\m{f})$ satisfies conditions analogous to those in Lemma \ref{lemma 10}. A detailed discussion is omitted here for brevity and is provided in the supplementary material. We summarize these conclusions in the following theorems.
\begin{theorem}[Multidimensional Rank Information]\label{Thm 11}
	Under the low-rank model, we have 
	\begin{equation}
		\begin{aligned}
			\m{\Theta}_{\m{\theta}_{\m{F}}\m{\theta}_{\m{F}}}=\m{\Theta}_{\m{\theta}_{\m{F}}\m{\theta}_{\m{F}}}^{\textnormal{LR}}. \label{5.2.9}
		\end{aligned}
	\end{equation}
	When $R<\min\{K,L\}$, we have $\textnormal{tr}\left(\m{\Theta}_{\m{\theta}_{\m{S}}\m{\theta}_{\m{S}}}- \m{\Theta}_{\m{\theta}_{\m{S}}\m{\theta}_{\m{S}}}^{\textnormal{LR}}\right)>0$.
\end{theorem}
\begin{theorem}[Multidimensional Coherence Structure]\label{Thm 12}
	Suppose $D\geq2$ and there exists $d$ such that $P_{d}>1$. When the steering vector $\m{a}(\m{f})$ is analytic and satisfies $\m{a}(\m{f})$ spans $\mathbb{C}^N$, and  there exists an $\m{f}$ such that $\m{b}^{1}(\m{f})$, or $\m{b}^{2}(\m{f})$, is not proportional to $\m{a}(\m{f})$, for any pairwise-incoherent $\m{C}$ and generic $\left(\m{f},\m{\rho}\right)$, we have
		\begin{equation}
		\begin{aligned}
			\m{\Theta}_{\m{\theta}_{\m{F}}\m{\theta}_{\m{F}}}\succeq \m{\Theta}_{\m{\theta}_{\m{F}}\m{\theta}_{\m{F}}}^{\textnormal{Cor}},\quad \textnormal{tr}\left(\m{\Theta}_{\m{\theta}_{\m{F}}\m{\theta}_{\m{F}}}- \m{\Theta}_{\m{\theta}_{\m{F}}\m{\theta}_{\m{F}}}^{\textnormal{Cor}}\right)>0. \label{5.2.10}
		\end{aligned}
	\end{equation} 
\end{theorem}
Theorem \ref{Thm 12} shows that exploiting the coherence structure can strictly reduce the CRB for the frequency matrix $\m{F}$. Nevertheless, in some applications, it is also important to examine the CRBs for the individual components of $\m{F}$, such as $\m{f}^1$ and $\m{f}^2$. For instance, in two-dimensional DOA estimation, these two components may correspond to the azimuth and elevation angles, whose estimation accuracies can be of separate interest. Due to the coupling between $\m{f}^1$ and $\m{f}^2$, Proposition \ref{prop2} is insufficient to resolve this issue. Therefore, we establish the following proposition.
\begin{proposition}\label{prop 13}
	Consider $\m{f}^{1},\m{f}^{2}\in\mathbb{R}^{K},\m{\alpha}\in\mathbb{R}^{n}$ as the parameters to be estimated, with its FIM given by
	\begin{equation}
		\begin{aligned}
			\begin{bmatrix}
				\m{I}_{\m{f}^{1}\m{f}^{1}}&\m{I}_{\m{f}^{1}\m{f}^{2}}&\m{I}_{\m{f}^{1}\m{\alpha}} \\
					\m{I}_{\m{f}^{2}\m{f}^{1}}&\m{I}_{\m{f}^{2}\m{f}^{2}}&\m{I}_{\m{f}^{2}\m{\alpha}}\\
						\m{I}_{\m{\alpha}\m{f}^{1}}&\m{I}_{\m{\alpha}\m{f}^{2}}&\m{I}_{\m{\alpha}\m{\alpha}}
			\end{bmatrix} \label{5.2.11}
		\end{aligned}
	\end{equation}
	being positive definite. Let $\m{\Theta}$ denote the CRB matrix, and partition it in the same manner as in \eqref{5.2.11}. Let $\m{\Theta}^{\textnormal{cons}}$ denote the constrained CRB matrix under the constraint shown in Proposition \ref{prop2}, we have
		$	\m{\Theta}_{\m{f}^{1}\m{f}^{1}}\succeq \m{\Theta}_{\m{f}^{1}\m{f}^{1}}^{\textnormal{cons}}, $
	with equality if and only if 
	\begin{equation}
		\begin{aligned}
			\textnormal{col}\Bigg(\Big(\m{I}_{\m{\alpha}\m{f}^{2}}\m{I}_{\m{f}^{2}\m{f}^{2}}^{-1}&\m{I}_{\m{f}^{2}\m{\alpha}}-\m{I}_{\m{\alpha}\m{\alpha}}\Big)^{-1} \\&
			\left(\m{I}_{\m{\alpha}\m{f}^{2}}\m{I}_{\m{f}^{2}\m{f}^{2}}^{-1}\m{I}_{\m{f}^{2}\m{f}^{1}}-\m{I}_{\m{\alpha}\m{f}^{1}}\right)\Bigg)\subseteq \textnormal{col}(\m{J}). \label{5.2.13}
		\end{aligned}
	\end{equation}
	The conclusion for $\m{\Theta}_{\m{f}^{2}\m{f}^{2}}$ follows analogously by interchanging $\m{f}^{1}$ and $\m{f}^{2}$.
\end{proposition}
The proof of Proposition \ref{prop 13} relies on the observation that the two matrices in \eqref{5.2.13} can be viewed as the equivalent FIM blocks associated with $\m{\alpha}$ and $\m{f}^{1}$ after eliminating $\m{f}^{2}$. Hence, Proposition \ref{prop 13} can be obtained by applying Proposition \ref{prop2} to the resulting FIM blocks. A detailed proof is provided in the supplementary material.\par
Applying Proposition \ref{prop 13}, we show that for UPAs, the steering vector of which is
\begin{equation}
	\begin{aligned}
		\m{a}(\m{f})=\m{a}_{1}(f^{1})\otimes \m{a}_{2}(f^{2}), \label{5.2.14}
	\end{aligned}
\end{equation}
with $\m{a}_{1}(f^{1})=\left[1,e^{i2\pi f^{1}},...,e^{i2\pi (N_{1}-1)f^{1}}\right]^{T}$, $\m{a}_{2}(f^{2})=\left[1,e^{i2\pi f^{2}},...,e^{i2\pi (N_{2}-1)f^{2}}\right]^{T}$ and $N=N_{1}N_{2}$, the CRBs for $\m{f}^{1}$ and $\m{f}^{2}$ can be strictly reduced by exploiting the coherence structure separately.
\begin{corollary}\label{Cor 14}
	Under the same assumptions as in Theorem \ref{Thm 4}, for UPAs, any pairwise-incoherent $\m{C}$, and generic $\left(\m{f},\m{\rho}\right)$ with $K<\min\{N_{1},N_{2}\}$, we have 
	\begin{equation}
		\begin{aligned}
			\textnormal{tr}\left(\m{\Theta}_{\m{f}^{1}\m{f}^{1}}- \m{\Theta}_{\m{f}^{1}\m{f}^{1}}^{\textnormal{Cor}}\right)>0, \quad \textnormal{tr}\left(\m{\Theta}_{\m{f}^{2}\m{f}^{2}}- \m{\Theta}_{\m{f}^{2}\m{f}^{2}}^{\textnormal{Cor}}\right)>0. \label{5.2.15}
		\end{aligned}
	\end{equation}
\end{corollary}
The proof proceeds by constructing a specific set of parameters under which \eqref{5.2.13} does not hold. Specifically, we choose $f_{k}^{1}=\frac{k-1}{N_{1}}$ and $\m{f}^{2}=\m{0}$. We provide a detailed proof of Corollary \ref{Cor 14} in the supplementary material.

\section{Numerical Validation} \label{validation}
In this section, we validate our theoretical analysis through numerical experiments.
Two array geometries are considered in the simulations. The first one is a nested array whose sensor position indices are $\{1,2,3,4,5,6,7,8,16,24,32,40\}$. The second one is a UPA with dimensions $N_{1}\times N_{2}=10 \times 8$. We denote the three CRBs under the general, low-rank, and coherence models by CRB-G, CRB-LR, and CRB-C, respectively. For the nested array, the frequencies are randomly generated with a minimum separation of $0.5/N$, where $N=40$. For the UPA, the frequencies $\m{f}^{1}$ and $\m{f}^{2}$ are randomly generated with a minimum separation of $0.5/(N_{1}\times N_{2})$. The magnitudes of the coherent coefficients are uniformly drawn from $[0.5,1]$, with their phases generated randomly. The source $\m{C}$ are generated from complex Gaussian distributions. Note that the complex Gaussian generation of $\m{C}$ is adopted only for simulation convenience and is not required by the theoretical analysis. The SNR is defined as
$\textnormal{SNR}=10\log_{10}\left(\frac{\|\m{X}\|_{F}^{2}}{ML\sigma}\right)$,
where $M$ is the number of sensors, $L$ is the number of snapshots, fixed at $20$, and $\sigma$ denotes the noise variance. We consider the following three cases. In the first case, $K=6$ and $D=3$, with group sizes $3$, $2$, and $1$. In the second case, $K=6$ and $D=1$, i.e., all sources are coherent. In the third case, $K=D=6$, i.e., all sources are mutually incoherent. The number of Monte Carlo trials is fixed at 500. \par
\begin{figure}[htbp] 
	\centering 
	\includegraphics[width=0.9\linewidth]{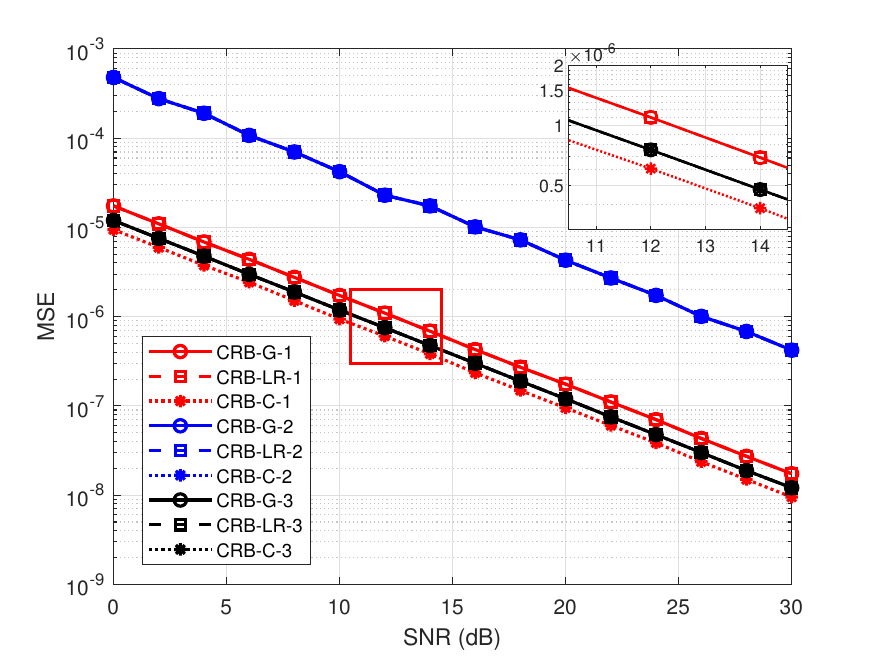} 
	\caption{Comparison of CRB-G, CRB-LR, and CRB-C versus SNR under the nested array.} 
	\label{fig2} 
\end{figure}
In Fig. \ref{fig2}, we present the results obtained with the nested array. In Scenario 1, represented by the red curves, we observe that CRB-G and CRB-LR coincide, and are both higher than CRB-C. For Scenario 2, represented by the blue curves, the condition in Theorem \ref{Thm 4} is not satisfied. In this case, the coherence model degenerates to the low-rank model, and the three CRBs coincide. For Scenario 3, represented by the black curves, the coherence model degenerates to the general model, and all three CRBs are identical. Table \ref{Table 2} compares the numerical values of the three CRBs at $10$ dB under scenario 1, where the relative differences are defined as 
\begin{equation}
	\begin{aligned}
		\Delta_{\textnormal{G,LR}}=\frac{|\textnormal{CRB}_{\textnormal{G}}-\textnormal{CRB}_{\textnormal{LR}}|}{\textnormal{CRB}_{\textnormal{G}}}, \quad \Delta_{\textnormal{G,C}}=\frac{|\textnormal{CRB}_{\textnormal{G}}-\textnormal{CRB}_{\textnormal{C}}|}{\textnormal{CRB}_{\textnormal{G}}}. \label{6.1}
	\end{aligned}
\end{equation}
The results verify the conclusions of Theorems \ref{Thm 3} and \ref{Thm 4} up to numerical precision. \par
Fig. \ref{fig3} presents the results obtained with the UPA. The corresponding numerical comparisons are also reported in Table \ref{Table 2}. These results are consistent with those obtained for the nested array, thereby validating the conclusions of Theorems \ref{Thm 11} and \ref{Thm 12} and Corollary \ref{Cor 14}.

\begin{table}[t]
	\centering
	\caption{Relative differences between the three CRBs.}
	\label{Table 2}
	\renewcommand{\arraystretch}{1.2}
	\setlength{\tabcolsep}{5pt}
	\begin{tabular}{c ccc}
		\toprule
		
		& $\m{f}$ under nested array
		& $\m{f}^{1}$ under UPA
		& $\m{f}^{2}$ under UPA \\
		\midrule
		\(\Delta_{\textnormal{G, LR}}\)
		& 6.1246$\times$ $10^{-16}$ 
		& 7.2437$\times$ $10^{-13}$  
		& 8.4474$\times$ $10^{-13}$  \\
		
		\(\Delta_{\textnormal{G, C}}\)
		& 0.4569 
		& 0.7860 
		& 0.9533 \\
		\bottomrule
	\end{tabular}
\end{table}

\begin{figure}[t]
	\centering
	\subfloat{
		\includegraphics[width=0.9\linewidth]{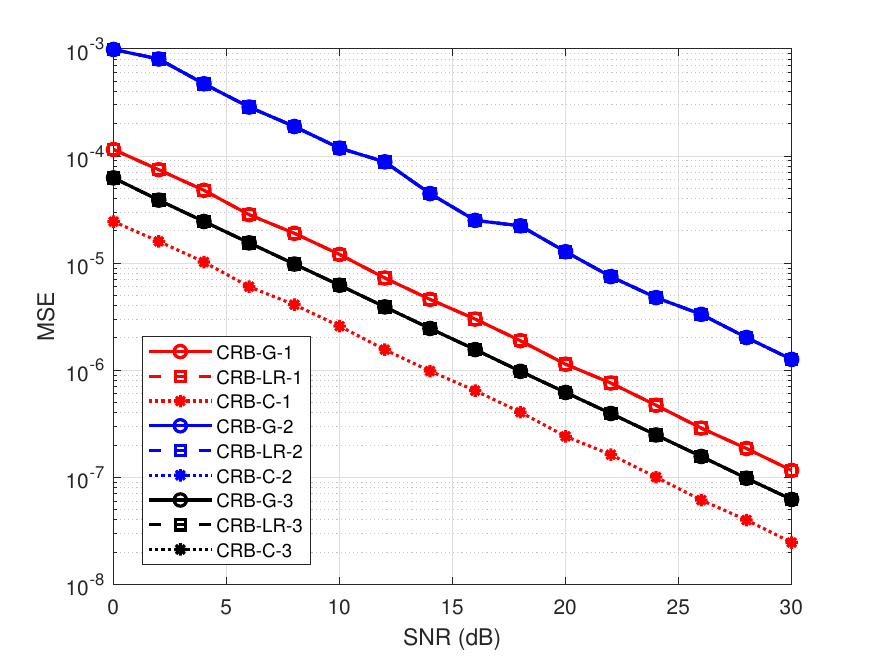}
		\label{fig3.1}
	}	
	\vspace{2mm}
	\subfloat{
		\includegraphics[width=0.9\linewidth]{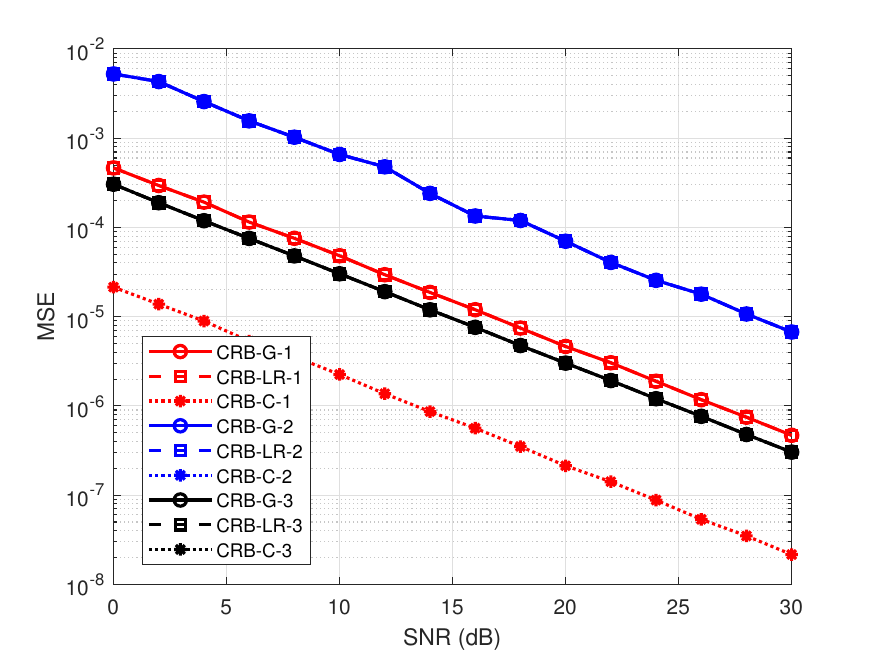}
		\label{fig3.2}
	}
	\caption{Comparison of CRB-G, CRB-LR, and CRB-C for $\m{f}^{1}$ (top) and $\m{f}^{2}$ (bottom) versus SNR under the UPA.}
	\label{fig3}
\end{figure}

\section{Conclusion}
In this paper, we studied the deterministic CRB for coherent DOA estimation under three models: the general model, the low-rank model, and the coherence model. We rigorously established the relationships among the three CRBs. In particular, we proved that exploiting the low-rankness of the source matrix alone does not reduce the frequency CRBs, whereas a strictly lower CRB can be achieved by further exploiting the coherent source structure. These results provide a theoretical explanation for the potential performance gain brought by source structure. Future work will study the CRB behavior of individual frequency components. Another direction is to extend the multidimensional analysis to array geometries beyond the UPA.

\appendices
\section{Proof of Proposition \ref{prop2}}\label{Appendix A}
To apply Theorem \ref{Thm 1}, we define  $\tilde{\m{g}}(\m{\psi},\m{\varphi})=\left[\m{\psi}^{T},\m{g}(\m{\varphi})^{T}\right]^{T}$
whose Jacobian matrix $\tilde{\m{J}}=\textnormal{blkdiag}(\m{I}_{K},\m{J})$ is of full column rank. The parametric constraint is then equivalent to 
\begin{equation}
	\begin{aligned}
		\left(\m{f},\m{\alpha}\right)\in\left\{\left(\m{f},\m{\alpha}\right):\left[\m{f}^{T},\m{\alpha}^{T}\right]^{T}=\tilde{\m{g}}(\m{\psi},\m{\varphi})\textnormal{ for some }\m{\psi},\m{\varphi}\right\}.\label{A.1}
	\end{aligned}
\end{equation} 
Let $\m{U}\in\mathbb{R}^{(n-m)\times n}$ have orthonormal rows and satisfy $\m{U}\m{J}=\m{0}$. Then we obtain that
\begin{equation}
	\begin{aligned}
		\tilde{\m{U}}=\begin{bmatrix}
			\m{0}_{n-m,K} &\m{U}
		\end{bmatrix}. \label{A.2}
	\end{aligned}
\end{equation}
also has orthonormal rows and satisfies $\tilde{\m{U}}\tilde{\m{J}}=\m{0}$. Define $\m{\Upsilon}=\m{U}^{T}\left(\m{U}\m{\Theta}_{\m{\alpha}\m{\alpha}}\m{U}^{T}\right)^{-1}\m{U}$. Applying Theorem \ref{Thm 1} yields that 
\begin{equation}
	\begin{aligned}
		\m{\Theta}^{\textnormal{cons}}=\m{\Theta}-\m{\Theta}\tilde{\m{U}}^{T}\left(\tilde{\m{U}}\m{\Theta}\tilde{\m{U}}^{T}\right)^{-1}\tilde{\m{U}}\m{\Theta},\label{A.3}
	\end{aligned}
\end{equation}
which deduces that
\begin{equation}
	\begin{aligned}
		\m{\Theta}^{\textnormal{cons}}_{\m{f}\m{f}}&=\m{\Theta}_{\m{f}\m{f}}-\m{\Theta}_{\m{f}\m{\alpha}}\m{\Upsilon}\m{\Theta}_{\m{\alpha}\m{f}},\\
		\m{\Theta}^{\textnormal{cons}}_{\m{\alpha}\m{\alpha}}&=\m{\Theta}_{\m{\alpha}\m{\alpha}}-\m{\Theta}_{\m{\alpha}\m{\alpha}}\m{\Upsilon}\m{\Theta}_{\m{\alpha}\m{\alpha}}.
		\label{A.4}
	\end{aligned}
\end{equation}
Since $\m{\Theta}_{\m{f}\m{\alpha}}\m{\Upsilon}\m{\Theta}_{\m{\alpha}\m{f}}\succeq\m{0}_{K}$, we conclude that \eqref{3.7} holds, with equality if and only if
\begin{equation}
	\begin{aligned}
		&\m{\Theta}_{\m{f}\m{\alpha}}\m{\Upsilon}\m{\Theta}_{\m{\alpha}\m{f}}\\
		=&\m{\Theta}_{\m{f}\m{\alpha}}\m{U}^{T}\left(\m{U}\m{\Theta}_{\m{\alpha}\m{\alpha}}\m{U}^{T}\right)^{-1}\m{U}\m{\Theta}_{\m{\alpha}\m{f}}\\
		=&\m{0}_{K}.\label{A.7}
	\end{aligned}
\end{equation}
Since the FIM given in \eqref{3.4} is positive definite and $\m{U}$ is row orthogonal, we get that $\m{U}\m{\Theta}_{\m{\alpha}\m{\alpha}}\m{U}^{T}\succ \m{0}_{n-m}$. Therefore, (\ref{A.7}) is equivalent to
\begin{equation}
	\begin{aligned}
		\m{U}\m{\Theta}_{\m{\alpha}\m{f}}=\m{0}_{n-m,K},\label{A.8}
	\end{aligned}
\end{equation}
which is further equivalent to
\begin{equation}
	\begin{aligned}
		\textnormal{col}(\m{\Theta}_{\m{\alpha}\m{f}})\subseteq \textnormal{col}(\m{J}).\label{A.12}
	\end{aligned}
\end{equation}
According to the Schur complement theorem \cite{zhang2006schur}, we have
\begin{equation}
	\begin{aligned}
		\m{\Theta}_{\m{\alpha}\m{f}}=-\m{I}_{\m{\alpha}\m{\alpha}}\m{I}_{\m{\alpha}\m{f}}\textnormal{Sch}^{-1}(\m{I},\m{I}_{\m{\alpha}\m{\alpha}}).\label{A.13}
	\end{aligned}
\end{equation}
The conclusion \eqref{3.8} is obtained by noting that $\textnormal{Sch}(\m{I},\m{I}_{\m{\alpha}\m{\alpha}})$ is positive definite, due to the positive definiteness of $\m{I}$. Since $\m{\Theta}_{\m{\alpha\alpha}}\succ \m{0}_{n}$, we obtain that $\m{\Theta}_{\m{\alpha\alpha}}\m{U}^{T}\neq \m{0}_{n,n-m}$. Therefore, we get that $\m{\Theta}_{\m{\alpha}\m{\alpha}}\m{\Upsilon}\m{\Theta}_{\m{\alpha}\m{\alpha}}\neq\m{0}_{n}$, and thus \eqref{3.6} holds, which completes the proof.
\section{Proof of Lemma \ref{Lem5}}\label{Appendix B}
The right side of (\ref{3.16a}) equals
\begin{equation}
	\begin{aligned}
		&\begin{bmatrix}
			\Re(\m{A})&-\Im(\m{A})\\
			\Im(\m{A})&\Re(\m{A})
		\end{bmatrix}\cdot
		\begin{bmatrix}
			\Re(\m{B})&-\Im(\m{B})\\
			\Im(\m{B})&\Re(\m{B})
		\end{bmatrix}\\
		=&\begin{bmatrix}
		    \Re(\m{A})\Re(\m{B})-\Im(\m{A})\Im(\m{B})&-\Re(\m{A})\Im(\m{B})-\Im(\m{A})\Re(\m{B})\\
		    \Re(\m{A})\Im(\m{B})+\Im(\m{A})\Re(\m{B})&\Re(\m{A})\Re(\m{B})-\Im(\m{A})\Im(\m{B})
		\end{bmatrix} \\
		=&\begin{bmatrix}
		    \Re(\m{A}\m{B})&-\Im(\m{A}\m{B})\\
		    \Im(\m{A}\m{B})&\Re(\m{A}\m{B})
		\end{bmatrix},
		\label{B.1}
	\end{aligned}
\end{equation}
which is exactly $\mathcal{R}(\m{A}\m{B})$. Similarly, the right side of \eqref{3.16b} equals
\begin{equation}
	\begin{aligned}
		\begin{bmatrix}
			\Re(\m{A})&-\Im(\m{A})\\
			\Im(\m{A})&\Re(\m{A})
		\end{bmatrix}\cdot
		\begin{bmatrix}
			\Re(\m{B})\\
			\Im(\m{B})
		\end{bmatrix}
		=\begin{bmatrix}
			\Re(\m{A})\Re(\m{B})-\Im(\m{A})\Im(\m{B})\\
			\Re(\m{A})\Im(\m{B})+\Im(\m{A})\Re(\m{B})
		\end{bmatrix}. \label{B.2}
	\end{aligned}
\end{equation}
The right side of \eqref{B.2} is exactly $\mathcal{V}(\m{A}\m{B})$.
\section{Proof of Lemma \ref{Lem7}}\label{Appendix C}
If $\textnormal{col}(\m{A})\subseteq\textnormal{col}(\m{B})$, then there exists $\m{T}\in\mathbb{C}^{l\times n}$ such that $\m{A}=\m{B}\m{T}$. Hence 
\begin{equation}
	\begin{aligned}
		\mathcal{V}(\m{A})=\mathcal{V}(\m{B}\m{T})=\mathcal{R}(\m{B})\cdot\mathcal{V}(\m{T}), \label{C.1}
	\end{aligned}
\end{equation} 
deducing that $\textnormal{col}(\mathcal{V}(\m{A}))\subseteq\textnormal{col}(\mathcal{R}(\m{B}))$. Conversely, if $\textnormal{col}(\mathcal{V}(\m{A}))\subseteq\textnormal{col}(\mathcal{R}(\m{B}))$, then there exists $\m{W}\in\mathbb{R}^{2l\times n}$ such that $\mathcal{V}(\m{A})=\mathcal{R}(\m{B})\m{W}$. Partition $\m{W}$ as 
\begin{equation}
	\begin{aligned}
		\m{W}=\begin{bmatrix}
			\m{W}_{1}\\
			\m{W}_{2}
		\end{bmatrix} \label{C.2}
	\end{aligned}
\end{equation}
with $\m{W}_{1},\m{W}_{2}\in \mathbb{R}^{l\times n}$. Letting $\m{T}=\m{W}_{1}+i\m{W}_{2}$, we have $\m{W}=\mathcal{V}(\m{T})$, and thus 
\begin{equation}
	\begin{aligned}
		\mathcal{V}(\m{A})=\mathcal{R}(\m{B})\cdot\mathcal{V}(\m{T})=\mathcal{V}(\m{BT}). \label{C.3}
	\end{aligned}
\end{equation}
Since $\mathcal{V}$ is injective, it follows from \eqref{C.3} that $\m{A}=\m{BT}$, and therefore $\textnormal{col}(\m{A})\subseteq\textnormal{col}(\m{B})$, completing the proof.
\section{Proof of Lemma \ref{Lem8}}\label{Appendix D}
We define $f_{k}=\frac{k-1}{N}$ for $k=1,2,...,K$. It follows from \eqref{3.36} that 
\begin{equation}
	\begin{aligned}
		\left[\m{A}^{H}\m{A}\right]_{ij}=\begin{cases}
				N,  &i=j,\\
			0,  &i\neq j,
		\end{cases} \label{D.1}
	\end{aligned}
\end{equation}
thus $\m{A}^{H}\m{A}=N\m{I}_{K}$. For $i=j$, we have
\begin{equation}
	\begin{aligned}
		\left[\m{A}^{H}\m{B}\right]_{ii}=\sum_{n=1}^{N}i2\pi(n-1)=i\pi N(N-1)\neq0. \label{D.2}
	\end{aligned}
\end{equation}
For $i\neq j$, defining $\omega=e^{i2\pi/N}$, we have
\begin{equation}
	\begin{aligned}
		\left[\m{A}^{H}\m{B}\right]_{ij}=i2\pi\sum_{n=1}^{N}(n-1)\omega^{(j-i)(n-1)}=-\frac{i2\pi N}{1-\omega^{j-i}}, \label{D.3}
	\end{aligned}
\end{equation}
which is nonzero. In conclusion each element of the matrix $\m{P}$ determined by the constructed $\m{f}$ is nonzero, hence $\m{P}_{ij}\neq \m{0}$.
\section{Proof of Lemma \ref{lemma 9}}\label{Appendix E}
Let $\m{P}_{\mathcal{I}_{j}}\triangleq\m{P}_{:,\mathcal{I}_{j}}=\left(\m{A}^{H}\m{A}\right)^{-1}\m{A}^{H}\m{B}_{\mathcal{I}_{j}}$, which is the optimal solution to the following optimization problem:
\begin{equation}
	\begin{aligned}
		\min_{\m{P}_{\mathcal{I}_{j}}}\left\|\m{B}_{\mathcal{I}_{j}}-\m{A}\m{P}_{\mathcal{I}_{j}}\right\|_{F}^{2}.\label{E.1}
	\end{aligned}
\end{equation} 
We partition $\m{P}_{\mathcal{I}_{j}}$ into $\m{P}_{\mathcal{I}_{i},\mathcal{I}_{j}}$ and $\m{P}_{-\mathcal{I}_{i},\mathcal{I}_{j}}$. Then \eqref{E.1} is equivalently transformed as 
\begin{equation}
	\begin{aligned}
		\min_{\m{P}_{\mathcal{I}_{i},\mathcal{I}_{j}},\m{P}_{-\mathcal{I}_{i},\mathcal{I}_{j}}}\left\|\m{B}_{\mathcal{I}_{j}}-\m{A}_{\mathcal{I}_{i}}\m{P}_{\mathcal{I}_{i},\mathcal{I}_{j}}-\m{A}_{-\mathcal{I}_{i}}\m{P}_{-\mathcal{I}_{i},\mathcal{I}_{j}}\right\|_{F}^{2}.\label{E.2}
	\end{aligned}
\end{equation}
For a fixed $\m{P}_{\mathcal{I}_{i},\mathcal{I}_{j}}$, we first solve for $\m{P}_{-\mathcal{I}_{i},\mathcal{I}_{j}}$ and eliminate it, which yields
\begin{equation}
	\begin{aligned}
		\min_{\m{P}_{\mathcal{I}_{i},\mathcal{I}_{j}}}\left\|\m{\Pi}_{\m{A}_{-\mathcal{I}_{i}}}^{\perp}\left(\m{B}_{\mathcal{I}_{j}}-\m{A}_{\mathcal{I}_{i}}\m{P}_{\mathcal{I}_{i},\mathcal{I}_{j}}\right)\right\|_{F}^{2}. \label{E.3}
	\end{aligned}
\end{equation}
We now show that $\m{\Pi}_{\m{A}_{-\mathcal{I}_{i}}}^{\perp}\m{A}_{\mathcal{I}_{i}}$ is of full column rank. Suppose there exists $\m{x}$ such that $\m{\Pi}_{\m{A}_{-\mathcal{I}_{i}}}^{\perp}\m{A}_{\mathcal{I}_{i}}\m{x}=\m{0}$, we obtain that $\m{A}_{\mathcal{I}_{i}}\m{x}\in \textnormal{col}\left(\m{A}_{-\mathcal{I}_{i}}\right)$. Therefore, there exists $\m{y}$ such that
\begin{equation}
	\begin{aligned}
		\m{A}_{\mathcal{I}_{i}}\m{x}=\m{A}_{-\mathcal{I}_{i}}\m{y}.\label{E.4}
	\end{aligned}
\end{equation}
Since $\m{A}$ is of full column rank, we obtain that $\m{x}=\m{0},\m{y}=\m{0}$, deducing that $\m{\Pi}_{\m{A}_{-\mathcal{I}_{i}}}^{\perp}\m{A}_{\mathcal{I}_{i}}$ is also of full column rank. By solving \eqref{E.3}, we get that 
\begin{equation}
	\begin{aligned}
	    \m{P}_{\mathcal{I}_{i},\mathcal{I}_{j}}=\Big[\left(\m{\Pi}_{\m{A}_{-\mathcal{I}_{i}}}^{\perp}\m{A}_{\mathcal{I}_{i}}\right)^{H}&\left(\m{\Pi}_{\m{A}_{-\mathcal{I}_{i}}}^{\perp}\m{A}_{\mathcal{I}_{i}}\right)\Big]^{-1} \\
	     \quad &\left(\m{\Pi}_{\m{A}_{-\mathcal{I}_{i}}}^{\perp}\m{A}_{\mathcal{I}_{i}}\right)^{H}\m{\Pi}_{\m{A}_{-\mathcal{I}_{i}}}^{\perp}\m{B}_{\mathcal{I}_{j}},\label{E.5}
	\end{aligned}
\end{equation}
which is equal to \eqref{5.1.2} after rearrangement.
\section{Proof of Lemma \ref{lemma 10}}\label{Appendix F}
Following Lemma \ref{lemma 9}, we select an index $k\in\mathcal{I}_{j}$ and it suffice to show that there exists $\m{f}$ such that
\begin{equation}
	\begin{aligned}
		\m{A}_{\mathcal{I}_{i}}^{H}\m{\Pi}_{\m{A}_{-\mathcal{I}_{i}}}^{\perp}\m{b}_{k}\neq\m{0},\label{F.1}
	\end{aligned}
\end{equation}
where $\m{b}_{k}=\m{b}(f_{k})$, and $f_k$ is chosen so that $\m{b}_{k}$ is not proportional to $\m{a}(f_k)$, as guaranteed by the assumptions of Lemma \ref{lemma 10}. Let $\m{f}_{-\mathcal{I}_i}$ denote the entries of $\m{f}$ after excluding those indexed by $\mathcal{I}_{i}$. We first prove that, by appropriately choosing $\m{f}_{-\mathcal{I}_{i}}$, one can ensure that $\m{b}_k\notin\textnormal{col}(\m{A}_{-\mathcal{I}_i})$. Note that $f_{k}$ is one entry of $\m{f}_{-\mathcal{I}_{i}}$ since $i\neq j$. Since $\m{a}(f)$ spans $\mathbb{R}^{N}$, there exists $f\neq f_{k}$, such that 
\begin{equation}
	\begin{aligned}
		\m{a}(f)\notin \textnormal{col}\left([\m{a}(f_{k}),\m{b}_{k}]\right). \label{F.2}
	\end{aligned}
\end{equation}
Combining that $\m{b}_{k}$ is not proportional to $\m{a}(f_k)$, we obtain that 
\begin{equation}
	\begin{aligned}
		\m{b}_{k}\notin \textnormal{col}\left([\m{a}(f),\m{a}(f_{k})]\right).\label{F.3}
	\end{aligned}
\end{equation}
Therefore, the resulting $f$ can be selected as an entry of $\m{f}_{-\mathcal{I}_{i}}$. Repeating this argument, we can construct the entire subvector $\m{f}_{-\mathcal{I}_{i}}$ and thus $\m{A}_{-\mathcal{I}_{i}}$. It follows from $\m{b}_k\notin\textnormal{col}(\m{A}_{-\mathcal{I}_i})$ that 
\begin{equation}
	\begin{aligned}
		\m{\Pi}_{\m{A}_{-\mathcal{I}_{i}}}^{\perp}\m{b}_{k}\neq\m{0}. \label{F.4}
	\end{aligned}
\end{equation}
We next construct $\m{f}_{\mathcal{I}_i}$. We choose an index $o\in\mathcal{I}_i$ and let $f_o=f_k+\delta$, where $\delta>0$ will be specified later. It is calculated that
\begin{equation}
	\begin{aligned}
		&\m{a}^{H}(f_{o})\m{\Pi}_{\m{A}_{-\mathcal{I}_{i}}}^{\perp}\m{b}_{k}\\
		=&\left(\m{a}(f_{k})+\delta \m{b}_{k}\right)^{H}\m{\Pi}_{\m{A}_{-\mathcal{I}_{i}}}^{\perp}\m{b}_{k}+o(\delta)\quad  (\textnormal{Taylor expansion.})\\
		=&\m{a}^{H}(f_{k})\m{\Pi}_{\m{A}_{-\mathcal{I}_{i}}}^{\perp}\m{b}_{k}+\delta\cdot\m{b}_{k}^{H}\m{\Pi}_{\m{A}_{-\mathcal{I}_{i}}}^{\perp}\m{b}_{k}+o(\delta)\\
		=&\delta\cdot\m{b}_{k}^{H}\m{\Pi}_{\m{A}_{-\mathcal{I}_{i}}}^{\perp}\m{b}_{k}+o(\delta). \quad (\m{a}(f_{k})\in\textnormal{col}(\m{A}_{-\mathcal{I}_{i}}).) \label{F.5}
	\end{aligned}
\end{equation}
It follows from \eqref{F.4} that $\m{b}_{k}^{H}\m{\Pi}_{\m{A}_{-\mathcal{I}_{i}}}^{\perp}\m{b}_{k}>0$. Therefore, when $\delta$ is sufficiently small and nonzero, we have $\m{a}^{H}(f_{o})\m\Pi_{\m {A}_{-\mathcal{I}_{i}}}^{\perp}\m{b}_{k}\neq 0$, and thus \eqref{F.1} holds.
We then fix such a $\delta$, and choose the remaining entries of $\m{f}_{\mathcal{I}_{i}}$ arbitrarily, subject only to the requirement that they are distinct from the already selected frequencies. Consequently, the above construction yields an $\m{f}$ with the required property, which completes the proof.

\bibliographystyle{IEEEtran}
\bibliography{references}

\end{document}